\documentclass[aps, pre, twocolumn, amsmath, superscriptaddress,showkeys,showpacs]{revtex4-2}
\usepackage[bottom]{footmisc}
\usepackage{lipsum}
\usepackage{float}
\usepackage{graphicx}
\usepackage{algorithm}
\usepackage{amssymb, amsmath, amsfonts}
\usepackage[caption=false]{subfig} 
\usepackage{pgfplots} 
\usepackage{xcolor}
\usepackage{multibib}
\usepackage{mathtools}
\usepackage{bm}
\usepackage{amsopn}
\usepackage{xr}
\usepackage{hyperref}
\usepackage[normalem]{ulem} 
\usepackage{float}

\usepackage{graphicx,epstopdf}
\usepackage{xcolor, soul}
\usepackage{mathrsfs} 
\usepackage{fontenc}
\usepackage{ragged2e}
\usepackage{verbatim}
\newtheorem{theorem}{Theorem}
\newtheorem{proof}{proof}
\newtheorem{proposition}{Proposition}

\def\be{\begin{equation}}
\def\ee{\end{equation}}
\def\bea{\begin{eqnarray}}
\def\eea{\end{eqnarray}}
\def\bfg{\begin{figure}[H]}
\def\efg{\end{figure}}
\begin{document}

    \title{\color{black}{Selection Rules for Species Coexistence in a Hierarchical May–Leonard Model}}
\author{Rakesh Samanta}
\thanks{These authors contributed equally to this work.}
\affiliation{Raja Rammohun Roy Mahavidyalaya, Radhanagar, Hooghly 712406, West Bengal, India.}

\author{Shraosi Dawn}
\thanks{These authors contributed equally to this work.}
\affiliation{Center for Computational Natural Sciences and Bioinformatics, International Institute of Information Technology Hyderabad, Gachibowli, Hyderabad-500032, Telangana, India.}

\author{Sk Jahiruddin}
\affiliation{Sister Nibedita Govt. General Degree College for Girls, Hastings House, 20 B, Judges Court Road, Alipore, Kolkata - 700027, West Bengal, India.}

\author{Sirshendu Bhattacharyya}
\affiliation{Raja Rammohun Roy Mahavidyalaya, Radhanagar, Hooghly 712406, West Bengal, India.}

\author{Chittaranjan Hens}
\affiliation{Center for Computational Natural Sciences and Bioinformatics, International Institute of Information Technology Hyderabad, Gachibowli, Hyderabad-500032, Telangana, India.}
\author{Sayantan Nag Chowdhury}
\email{jcjeetchowdhury1@gmail.com}
\affiliation{School of Science, Constructor University, P.O. Box 750561, 28725 Bremen, Germany.}
\affiliation{Department of Mathematics, Nalanda University, Rajgir, Bihar 803116, India}

\begin{abstract}
One of the central challenges in evolutionary dynamics is understanding why some species combinations persist while others disappear. Although cyclic-interaction models have provided fundamental insights into biodiversity maintenance, much less is known about how hierarchical competitive interactions shape long-term community organization. Here, we investigate a hierarchical extension of the May-Leonard model, in which species interact through a directed predation chain while undergoing reproduction and mortality. Combining mean-field analysis with Monte Carlo simulations, we show that the fully coexisting state is generically unstable, causing the dynamics to evolve toward lower-dimensional coexistence states. The simulations further reveal stochastic extinctions dominating small populations with the dynamics progressively approaching the mean-field predictions as the system size increases. Rather than permitting arbitrary species combinations, the hierarchical-interaction structure dynamically constrains coexistence by selecting only specific subsets of species for long-term persistence. We show that these admissible coexistence states have a natural graph-theoretic interpretation as independent sets in the hierarchical interaction network, thereby providing general constraints on coexistence in hierarchical communities. Together, these results establish a theoretical framework linking hierarchical interactions, dynamical selection, graph topology, and biodiversity organization, extending the classical May-Leonard model beyond cyclic competition.

\end{abstract}

\maketitle

\section{Introduction}
\label{intro}

{\color{black} \par Understanding the mechanisms that govern biodiversity and species persistence is a central challenge in theoretical ecology and nonlinear dynamics. A fundamental question in coexistence theory is not simply whether coexistence is possible but which combinations of species are dynamically selected from the many mathematically feasible coexistence states. Because ecological communities are sustained by complex interaction networks, the structural organization of these networks plays a central role in determining long-term coexistence and ecosystem stability.

\par Ecological interaction networks can exhibit fundamentally different organizational structures. Among these, cyclic interaction networks have emerged as a paradigmatic framework for investigating how local competitive interactions generate rich collective dynamics. The rock-paper-scissors (RPS) model \cite{lotka1920analytical, volterra1927fluctuations, may1975nonlinear, he2010spatial, jiang2012multi, huang2023fitness, menezes2023spatial, cheng2014mesoscopic, park2017emergence, wang2022effect, avelino2012neummann} represents the canonical example of cyclic dominance, in which each species suppresses one competitor while being suppressed by another \cite{mohd2021interplay}. Within the Lotka-Volterra \cite{wangersky1978lotka, anisiu2014lotka, bunin2017ecological, din2013dynamics, takeuchi1996global, hernandez1997lotka} and May-Leonard \cite{serrao2017stochastic, may1975nonlinear, chi1998asymmetric, barendregt2023heteroclinic} formalisms, RPS dynamics has served as a fundamental model for exploring cyclic dominance, species coexistence, extinction transitions, and pattern formation \cite{leonard1975nonlinear, sinervo1996rock, szabo2007evolutionary, wu2010, claussen2008cyclic}. These models have successfully described diverse biological systems, including microbial consortia \cite{kerr2002local} and side-blotched lizards \cite{sinervo1996rock}. Incorporating ecological processes such as reproduction, mortality, and spatial mobility has further revealed that biodiversity can be highly sensitive to small variations in demographic parameters \cite{reichenbach2007mobility, reichenbach2008instability, mobiia2010jtb, bhattacharyya2020pre, islam2022pre, juul2012pre, szolnoki2009njp, serrao2021stabilizing}. More recently, several studies have examined the resilience of these systems to environmental and demographic perturbations through various notions of dynamical stability \cite{https://doi.org/10.1002/ecy.1601, https://doi.org/10.1111/ele.13457, baert2016per, de2016reintroducing, chatterjee2023response, chen2024stability, dai2025chaos}.

While cyclic interaction networks have provided fundamental insights into biodiversity maintenance, many ecological communities are not organized through closed competitive loops. Instead, interactions often exhibit hierarchical (non-cyclic or strictly transitive) structure \cite{reuter2010ecological, miller2008hierarchical, maia2024hierarchical, park2017emergence}, in which species can be ordered by competitive or predatory dominance and interactions propagate directionally without forming cycles. Such hierarchical organization breaks the symmetry of cyclic interactions by establishing a clear trophic gradient, creating an absolute top predator that experiences no predation and a bottom species that has no prey. Although idealized, hierarchical interaction structures arise naturally in a wide range of ecological settings, including dominance hierarchies, competition-colonization trade-offs, trophic ordering, and size-structured communities. Consequently, they provide a natural framework for investigating how directional interactions influence coexistence and the maintenance of biodiversity \cite{yang2025understanding}.

Despite the extensive literature on cyclic competitive systems \cite{szolnoki2014cyclic,chowdhury2021complex,chowdhury2023eco,chowdhury2021eco,roy2023time, park2013persistent, park2018multistability}, comparatively less attention has been devoted to understanding how hierarchical interaction structures determine which subsets of species persist and which are excluded from the long-term dynamics. In particular, it remains unclear how hierarchical interactions reorganize coexistence patterns and how transitions between different coexistence states emerge as ecological parameters vary. Recent work on hierarchical competition \cite{miller2025multispecies} has demonstrated that species-rich communities can coexist even when most species pairs cannot coexist in isolation, highlighting how community-level interactions can fundamentally alter coexistence beyond pairwise predictions. Earlier, Hastings \cite{hastings1980disturbance} showed that competition-colonization trade-offs can promote biodiversity by preventing superior competitors from monopolizing available space, allowing weaker competitors to persist through enhanced colonization. These studies suggest that hierarchical interactions can generate coexistence mechanisms that differ qualitatively from those operating in cyclic systems, motivating a systematic investigation of coexistence selection in hierarchical ecological networks.

Motivated by these developments, we investigate a hierarchical multi-species May-Leonard model in which species interact through directed predation chains rather than closed cyclic loops. Using a generalized May-Leonard formalism, we incorporate species-dependent predation, reproduction, and natural mortality rates, and analyze the resulting mean-field dynamics alongside Monte Carlo simulations. We systematically characterize the equilibrium structure, identify the mechanisms governing long-term persistence, and determine how stability is redistributed among competing coexistence states as demographic parameters vary. Our analysis demonstrates that hierarchical interactions dynamically constrain the feasible coexistence space, eliminating many mathematically admissible equilibrium configurations while selecting only specific subsets of species for long-term persistence. Furthermore, interpreting the hierarchical interaction network from a graph-theoretic perspective reveals general constraints on coexistence that naturally extend to larger hierarchical systems. Together, these results provide a unified framework for understanding coexistence selection, stability transitions, and biodiversity organization in hierarchical ecological communities, thereby extending the classical May-Leonard framework beyond cyclic competition.}

\begin{figure*}
    \centering
    \includegraphics[width=0.65\linewidth]{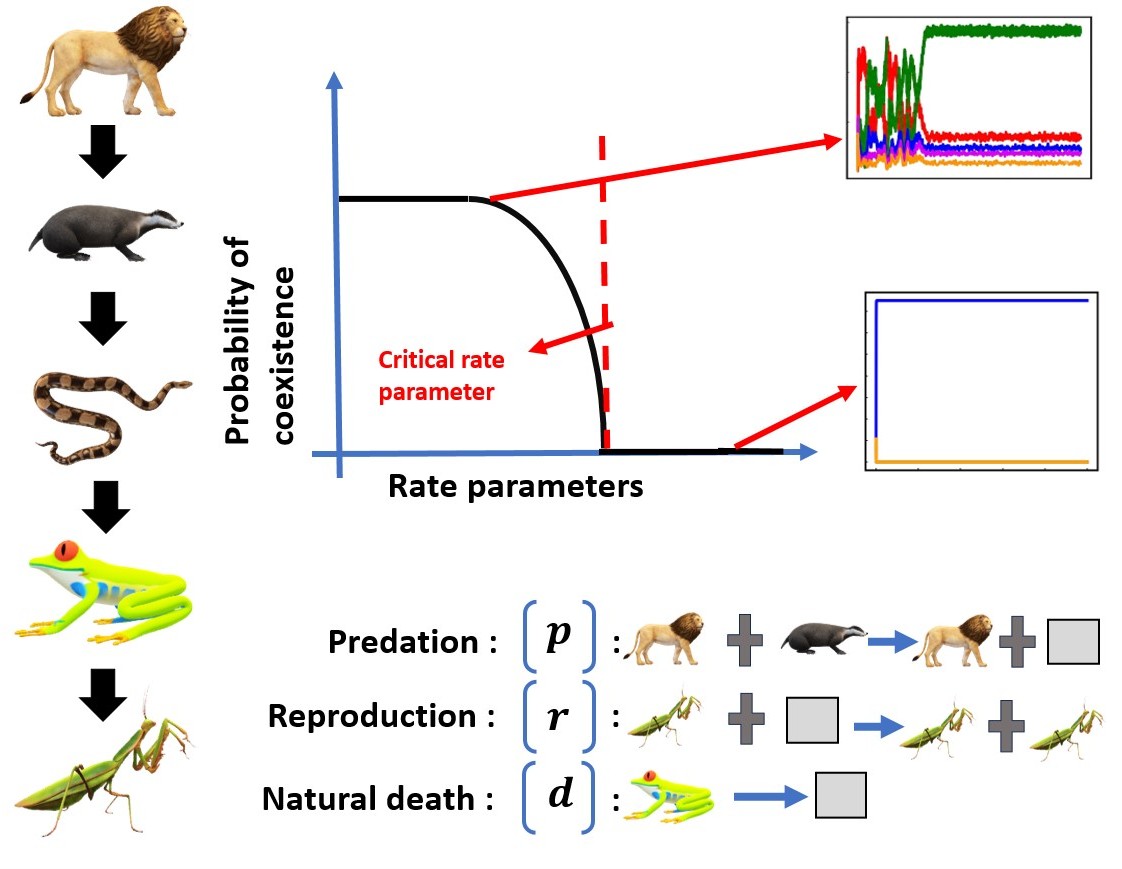}
    \caption{The schematic diagram illustrates all possible intra- and inter-species interactions among the five species considered in the model. It shows that when the existence probability is sufficiently high, all five species coexist in a stable state. However, as the rate parameter exceeds a critical threshold, the probability of coexistence  of all species decrease and eventually approach extinction. The empty square boxes denote vacant or unoccupied sites in the system, while the black arrows indicate the predator–prey interactions between species. The lion is placed at the highest trophic level, signifying that it has no natural predator within the interaction network. The animal icons are used solely as symbolic representations of the five species and do not correspond to any specific biological system or interaction.
  }
    \label{fig:placeholder}
\end{figure*}

{\color{black}

\section{Five-Species Hierarchical Model}}
\label{model}
\noindent
\subsection{Five-species model}

\noindent
We consider a hierarchical extension of the May–Leonard framework \cite{leonard1975nonlinear} consisting of five interacting species $A$, $B$, $C$, $D$ and $E$ which occupy sites in a well-mixed environment containing vacant locations $V$. {(\color{black}A simplified three-species hierarchical model exhibiting the same underlying mechanisms is analyzed in detail in Appendix~\ref{Three-species}.)} Each site can thus exist in one of six possible states: occupied by one of the five species or empty. The population densities of all states satisfy the conservation condition
$$\rho_a+\rho_b+\rho_c+\rho_d+\rho_e+\rho_v=1$$
The temporal dynamics of the system are driven by three fundamental ecological processes—predation, reproduction, and natural death $\--$ each characterized by specific rate constants.

\paragraph*{Predation hierarchy:}

Unlike the cyclic rock–paper–scissors interaction, the present model incorporates a directed hierarchical chain of predation. Species 
$A$ preys on $B$, $B$ preys on $C$, $C$ on $D$, and $D$ on $E$. The lowest species $E$ has no prey but is susceptible to predation by $D$. The interaction rules are expressed as
\bea
A + B & \xrightarrow{p_a} & A + V \nonumber \\
B + C & \xrightarrow{p_b} & B + V \nonumber \\
C + D & \xrightarrow{p_c} & C + V \nonumber \\
D + E & \xrightarrow{p_d} & D + V
\label{predation}
\eea
where $p_i$ denotes the predation rate of the $i$-th predator. Each predation event removes one prey individual and converts its site into a vacancy, representing consumption or displacement in an ecological sense. This directed topology breaks the cyclic symmetry typical of RPS systems and establishes a trophic gradient from species $A$ to species $E$.

\paragraph*{Reproduction:}

Each species can reproduce by colonizing a neighboring vacant site with a species-specific rate constant $r_x$:
\begin{equation}
    x + V \space \xrightarrow{r_x} \space 2x; \;\;\;\; x \;\epsilon \; \lbrace A, B, C, D, E \rbrace
\label{reproduction}
\end{equation}

This rule captures local reproduction or recruitment limited by resource availability, as reproduction requires the presence of an adjacent empty site.

\paragraph*{Natural mortality:}

Independently of predation, each species experiences spontaneous death with rate $d_x$:
\begin{equation}
    x \space \xrightarrow{d_x} \space V; \;\;\;\; x \;\epsilon \; \lbrace A, B, C, D, E \rbrace
\label{death}
\end{equation}

This process represents intrinsic mortality or environmental loss, transforming a living site into a vacancy without affecting other species directly.


\subsection{Mean-field representation}

{\color{black}The mean-field description assumes a well-mixed population and neglects spatial correlations, allowing the dynamics to be expressed solely in terms of species densities $\rho_i(t)$.} Under this approximation, the time evolution of each species is governed by mean-field rate equations.
The time evolution of each species results from the balance among reproduction, predation losses, and natural death, constrained by the fraction of vacant sites $\rho_v=1-\sum\limits_i \rho_i$ with $i=a,b,\cdots,e$ respectively. For the top predator $A$, which has no predator, the evolution equation reads
\bea
\dot{\rho_a} = \rho_a \left( r_a \rho_v - d_a \right)
\label{diff_a}
\eea
For each subsequent species $i=b,c,d,e$, the dynamics incorporate both intrinsic processes and predation from the immediate superior species:
\bea
\dot{\rho_i} = \rho_i \left( r_i \rho_v - p_{i-1} \rho_{i-1} - d_i \right)
\label{diff_all}
\eea

{\color{black}This formulation yields a hierarchically coupled nonlinear system that generalizes the May–Leonard equations from cyclic competition to asymmetric directional interactions.} The hierarchical coupling introduces a sequential dependency: species $i$ is directly affected only by its immediate predator $i-1$, while indirectly influenced by higher-level species through cascading effects.

\par {\color{black}The resulting mean-field system defines a hierarchical network of competitive interactions in which each species is directly influenced by its demographic parameters and by predation from its immediate superior competitor. The interplay among reproduction, mortality, and predation determines the admissible equilibrium configurations and their stability properties, which are analyzed in the following sections.}


{\color{black}
\section{Equilibrium points}

\par In contrast to conventional cyclic competition models, the equilibrium structure of the five-species hierarchy is significantly richer. Different states will be possible depending on the relative magnitudes of mortality, reproduction, and predation parameters, including extinctions, single-species survival, and coexistence involving different subsets of species. Since each species may either persist or become extinct, the full five-species system possesses up to $2^N=2^5=32$ equilibrium configurations. Rather than treating these equilibria independently, we classify them by the number and arrangement of surviving species, providing a natural framework for understanding coexistence selection in hierarchical competitive systems.

\subsection{Extinction equilibrium}

The simplest equilibrium configuration corresponds to the complete extinction of all species. In this state, the density of every species vanishes, and the system consists entirely of vacant sites. Biologically, this equilibrium represents the collapse of the ecological community due to overwhelming reproduction by mortality, i.e. 
$r_i < d_i \quad \forall i$. The corresponding equilibrium densities are given by
\begin{equation}
(\rho_a^*,\rho_b^*,\rho_c^*,\rho_d^*,\rho_e^*)=(0,0,0,0,0),
\end{equation}
with $\rho_v^*=1$.

\subsection{Single-species equilibria}

A second class of equilibria arises when only one species survives, while all others become extinct. These states correspond to complete competitive exclusion, in which a single species occupies all available habitat, except for the fraction required to balance mortality and reproduction. Since each of the five species may persist independently, the model admits five distinct single-species equilibria. The existence of such states requires a positive equilibrium density for the surviving species, which imposes the feasibility condition $r_i>d_i \quad i = a,b,c,d,e$, i.e., \, for each species $i$, the growth rate per-capita must exceed the mortality. Explicit expressions for the equilibrium densities are given by 
\begin{equation}
    \rho_i^* = 1 - \rho_v^* = 1 - \frac{d_i}{r_i}, \quad i = a,b,c,d,e.
\end{equation}

\subsection{Two-species coexistence equilibria}

Two-species coexistence states constitute the simplest nontrivial coexistence configurations supported by the hierarchical model. Depending on whether the surviving species are adjacent in the predation hierarchy, two qualitatively different equilibrium structures emerge.

\paragraph*{Consecutive species:}

When two adjacent species coexist, predation directly couples their dynamics, and the equilibrium is generally an isolated point in phase space. The admissible consecutive pairs are $(A,B)$, $(B,C)$, $(C,D)$, and $(D,E)$. For a generic consecutive pair $(i,i+1)$, the equilibrium densities are

\begin{equation}
(\rho_i^*,\rho_{i+1}^*)=
\left(
\frac{d_i r_{i+1}-d_{i+1} r_i}{p_i r_i},
-\frac{\begin{array}{@{}l@{}}
    d_i p_i + d_i r_{i+1} - d_{i+1} r_i \\
    {} - p_i r_i
  \end{array}}{p_i r_i}
\right),
\end{equation}

subject to the feasibility condition that both equilibrium densities remain positive.

\paragraph*{Nonadjacent species:}

A qualitatively different situation arises when the surviving species are separated in the hierarchy and therefore do not interact directly. In this case, coexistence is possible only when the mortality-to-reproduction ratios satisfy

\begin{equation} 
\frac{d_i}{r_i}=\frac{d_j}{r_j}.
\label{eq9}
\end{equation}

The equilibrium densities are not uniquely determined. Instead, they satisfy the conservation relation

\begin{equation} \label{eq10}
\rho_i^*+\rho_j^*=1-\frac{d_i}{r_i},
\end{equation}

which defines a continuous family of coexistence equilibria.

For the five-species hierarchy, the nonadjacent coexistence pairs are

(A, C),\quad
(A, D),\quad
(A, E),\quad
(B, D),\quad
(B, E),\quad
(C, E).

The corresponding coexistence manifolds are therefore

\begin{align}
\rho_a^*+\rho_c^* &= 1-\frac{d_a}{r_a},
& \frac{d_a}{r_a} &= \frac{d_c}{r_c}, \\
\rho_a^*+\rho_d^* &= 1-\frac{d_a}{r_a},
& \frac{d_a}{r_a} &= \frac{d_d}{r_d}, \\
\rho_a^*+\rho_e^* &= 1-\frac{d_a}{r_a},
& \frac{d_a}{r_a} &= \frac{d_e}{r_e}, \\
\rho_b^*+\rho_d^* &= 1-\frac{d_b}{r_b},
& \frac{d_b}{r_b} &= \frac{d_d}{r_d}, \\
\rho_b^*+\rho_e^* &= 1-\frac{d_b}{r_b},
& \frac{d_b}{r_b} &= \frac{d_e}{r_e}, \\
\rho_c^*+\rho_e^* &= 1-\frac{d_c}{r_c},
& \frac{d_c}{r_c} &= \frac{d_e}{r_e}.
\end{align}

Unlike consecutive pairs, these nonadjacent coexistence states form continuous equilibrium manifolds rather than isolated equilibrium points. These coexistence manifolds constitute an important feature of the hierarchical model and foreshadow the coexistence selection rules discussed later.

\subsection{Three-species coexistence equilibria}

The equilibrium structure becomes increasingly rich as three species coexist. Three-species coexistence states represent an intermediate level of biodiversity between pairwise coexistence and full community coexistence. Compared with two-species states, the equilibrium structure becomes considerably richer because the surviving species may form consecutive, mixed, or entirely nonadjacent configurations within the hierarchy. Of particular interest are coexistence states formed entirely by nonadjacent species, which occupy a distinguished role in the hierarchy. These configurations provide the first indication that hierarchical interactions constrain coexistence in a highly structured manner rather than allowing arbitrary combinations of surviving species.

\paragraph*{Consecutive species.}

When three adjacent species survive simultaneously, their densities are coupled through successive predation interactions. For a generic consecutive triplet $(i,i+1,i+2)$, the equilibrium densities are obtained from the steady-state conditions as

\begin{align}
\rho_i^* &= \frac{r_{i+1}\rho_v^* - d_{i+1}}{p_i}, \\
\rho_{i+1}^* &= \frac{r_{i+2}\rho_v^* - d_{i+2}}{p_{i+1}}, \\
\rho_{i+2}^* &= 1-\rho_v^*-\rho_i^*-\rho_{i+1}^*.
\end{align}

Thus, the coexistence densities are uniquely determined once the equilibrium vacancy density $\rho_v^*$ is specified. For the five-species hierarchy, the admissible consecutive triplets are $(A,B,C)$, $(B,C,D)$, and $(C,D,E)$. Their explicit equilibrium densities are

\begin{itemize}

\item \textbf{Species A, B \& C:}

\[
\begin{bmatrix}
\rho_a^* \\
\rho_b^* \\
\rho_c^* \\
\rho_d^* \\
\rho_e^*
\end{bmatrix}
=
\begin{bmatrix}
\dfrac{d_a r_b-d_b r_a}{p_a r_a} \\
\dfrac{d_a r_c-d_c r_a}{p_b r_a} \\
-\dfrac{
 \begin{array}{@{}l@{}}
    d_a p_a p_b + d_a p_a r_c + d_a p_b r_b \\
    {} - d_b p_b r_a - d_c p_a r_a - p_a p_b r_a
  \end{array}
}{p_a p_b r_a} \\
0 \\
0
\end{bmatrix}
\]

\item \textbf{Species B, C \& D:}

\[
\begin{bmatrix}
\rho_a^* \\
\rho_b^* \\
\rho_c^* \\
\rho_d^* \\
\rho_e^*
\end{bmatrix}
=
\begin{bmatrix}
0 \\
\dfrac{d_b r_c-d_c r_b}{p_b r_b} \\
\dfrac{d_b r_d-d_d r_b}{p_c r_b} \\
-\dfrac{
\begin{array}{@{}l@{}}
    d_b p_b p_c + d_b p_b r_d + d_b p_c r_c \\
    {} - d_c p_c r_b - d_d p_b r_b - p_b p_c r_b
  \end{array}
}{p_b p_c r_b} \\
0
\end{bmatrix}
\]

\item \textbf{Species C, D \& E:}

{\footnotesize
\[
\begin{bmatrix}
\rho_a^* \\
\rho_b^* \\
\rho_c^* \\
\rho_d^* \\
\rho_e^*
\end{bmatrix}
=
\begin{bmatrix}
0 \\
0 \\
\dfrac{d_c r_d-d_d r_c}{p_c r_c} \\
\dfrac{d_c r_e-d_e r_c}{p_d r_c} \\
-\dfrac{
d_c p_c p_d+d_c p_c r_e+d_c p_d r_d
-d_d p_d r_c-d_e p_c r_c-p_c p_d r_c
}{p_c p_d r_c}
\end{bmatrix}
\]
}



\end{itemize}

\paragraph*{Mixed configurations.}

Beyond consecutive coexistence states, the hierarchy also permits configurations containing both adjacent and nonadjacent species. In these cases, the densities of adjacent species are constrained by predation, while the nonadjacent species introduce additional coexistence degrees of freedom. Consequently, these equilibria generally form coexistence manifolds rather than isolated points. The admissible mixed configurations are

(A,B,D),\quad
(A,B,E),\quad
(A,C,D),\quad
(A,D,E),\quad
(B,C,E),\quad
(B, D,E).

The corresponding coexistence relations are

\begin{align}
\rho_a^* &= \frac{d_a r_b-d_b r_a}{p_a r_a},
&
\rho_a^*+\rho_b^*+\rho_d^*
&=
1-\frac{d_a}{r_a},
&
\frac{d_a}{r_a}
&=
\frac{d_d}{r_d},
\\
\rho_a^* &= \frac{d_a r_b-d_b r_a}{p_a r_a},
&
\rho_a^*+\rho_b^*+\rho_e^*
&=
1-\frac{d_a}{r_a},
&
\frac{d_a}{r_a}
&=
\frac{d_e}{r_e},
\\
\rho_c^* &= \frac{d_c r_d-d_d r_c}{p_c r_c},
&
\rho_a^*+\rho_c^*+\rho_d^*
&=
1-\frac{d_c}{r_c},
&
\frac{d_a}{r_a}
&=
\frac{d_c}{r_c},
\\
\rho_d^* &= \frac{d_d r_e-d_e r_d}{p_d r_d},
&
\rho_a^*+\rho_d^*+\rho_e^*
&=
1-\frac{d_d}{r_d},
&
\frac{d_a}{r_a}
&=
\frac{d_d}{r_d},
\\
\rho_b^* &= \frac{d_b r_c-d_c r_b}{p_b r_b},
&
\rho_b^*+\rho_c^*+\rho_e^*
&=
1-\frac{d_b}{r_b},
&
\frac{d_b}{r_b}
&=
\frac{d_e}{r_e},
\\
\rho_d^* &= \frac{d_d r_e-d_e r_d}{p_d r_d},
&
\rho_b^*+\rho_d^*+\rho_e^*
&=
1-\frac{d_d}{r_d},
&
\frac{d_b}{r_b}
&=
\frac{d_d}{r_d}.
\end{align}

\paragraph*{Nonadjacent species.}

A particularly important coexistence state arises when all surviving species are mutually nonadjacent in the hierarchy. For the five-species model, the only such configuration is $(A, C, E)$. Since no pair of surviving species interacts directly, coexistence is possible provided

\begin{equation} \label{eq26}
\frac{d_a}{r_a} = \frac{d_c}{r_c} = \frac{d_e}{r_e},
\end{equation}

with equilibrium densities satisfying

\begin{equation} \label{eq27}
\rho_a^*+\rho_c^*+\rho_e^*=1-\frac{d_c}{r_c}.
\end{equation}

This relation defines a two-dimensional coexistence manifold rather than an isolated equilibrium point. Among the coexistence states of all three-species, the configuration $(A, C ,E)$ is distinguished by the complete absence of direct predatory interactions between the surviving species. As will be shown later, this state plays a central role in the coexistence-selection mechanism of the hierarchical model.

\subsection{Four-species coexistence equilibria}

Four-species coexistence states correspond to configurations in which only a single species is absent from the community. These equilibria represent the highest level of partial coexistence supported by the hierarchy and form an intermediate stage between three-species coexistence and the fully populated interior state.

For a generic four-species configuration $(i,i+1,i+2,i+3)$, the equilibrium densities are given by

{\color{black}
\begin{align}
\rho_i^* &= \frac{r_{i+1}\rho_v^*-d_{i+1}}{p_i}, \\
\rho_{i+1}^* &= \frac{r_{i+2}\rho_v^*-d_{i+2}}{p_{i+1}}, \\
\rho_{i+2}^* &= \frac{r_{i+3}\rho_v^*-d_{i+3}}{p_{i+2}}, \\
\rho_{i+3}^* &= 1-\rho_v^*-\rho_i^*-\rho_{i+1}^*-\rho_{i+2}^*.
\end{align}}

The five-species hierarchy admits five possible four-species configurations,

(A, B ,C,D),\quad
(A, B,C,E),\quad 
(A,B,D,E),\quad
(A,C,D,E),\quad
(B,C,D,E). 

The feasibility of these states requires that all equilibrium densities remain positive. However, the hierarchical structure imposes strong constraints on the existence of four-species coexistence. In particular, the symmetric version of the model considered later does not support feasible four-species coexistence states. Consequently, these equilibria serve primarily as transitional configurations linking lower-dimensional coexistence states to the fully populated interior equilibrium. Note that, although four-species equilibria formally exist, they are not dynamically relevant in the parameter regimes investigated later in Sec.\ \ref{numerical}.

\subsection{Interior coexistence equilibrium}

The interior coexistence equilibrium corresponds to the state in which all five species persist simultaneously with positive densities. This equilibrium represents the maximum level of biodiversity permitted by the deterministic mean-field description and therefore serves as a natural reference point for assessing the hierarchy's long-term coexistence capacity.

The resulting equilibrium densities of the interior equilibrium for a generic five-species system are

\begin{align}
S^* &= \sum_{i=a}^{e}\rho_i^*=1-\frac{d_a}{r_a},
\\
\rho_a^* &=
\frac{r_b(d_a/r_a)-d_b}{p_a},
\\
\rho_b^* &=
\frac{r_c(d_a/r_a)-d_c}{p_b},
\\
\rho_c^* &=
\frac{r_d(d_a/r_a)-d_d}{p_c},
\\
\rho_d^* &=
\frac{r_e(d_a/r_a)-d_e}{p_d},
\\
\rho_e^* &=
\left(1-\frac{d_a}{r_a}\right)
-(\rho_a^*+\rho_b^*+\rho_c^*+\rho_d^*).
\end{align}

The interior equilibrium is feasible only when all species densities remain positive, i.\ e.\ ,

$$
\rho_i^*>0,
\qquad
i=a,b,c,d,e.$$

These positivity conditions impose constraints on the demographic and predation parameters and determine the region of parameter space in which full coexistence is admissible. The stability of the interior coexistence equilibrium is examined in the following section.

\section{Instability of the Interior Coexistence State}

The existence of an interior coexistence equilibrium does not necessarily imply long-term species persistence. For ecological coexistence to be dynamically realized, the equilibrium must remain stable against small perturbations in species densities. We therefore investigate the linear stability of the full coexistence state derived in the previous section. By evaluating the Jacobian matrix at the interior equilibrium and applying the Routh--Hurwitz criterion, we demonstrate that the interior coexistence state is generically unstable for positive demographic and predation parameters. Consequently, complete coexistence cannot serve as an asymptotically stable attractor of the hierarchical dynamics.

\subsection{Jacobian Matrix}

Let $\boldsymbol{\rho}^{*}=(\rho_a^*,\rho_b^*,\rho_c^*,\rho_d^*,\rho_e^*)$ denotes the interior equilibrium. The local stability of this state is determined by the Jacobian matrix. Evaluating the partial derivatives of the mean-field equations at the interior equilibrium yields

\[
\scriptsize
\begin{array}{c}
J^* = \\[2mm]
-\begin{pmatrix}
r_a \rho_a^* & r_a \rho_a^* & r_a \rho_a^* & r_a \rho_a^* & r_a \rho_a^* \\
\rho_b^*(r_b+p_a) & r_b \rho_b^* & r_b \rho_b^* & r_b \rho_b^* & r_b \rho_b^* \\
r_c \rho_c^* & \rho_c^*(r_c+p_b) & r_c \rho_c^* & r_c \rho_c^* & r_c \rho_c^* \\
r_d \rho_d^* & r_d \rho_d^* & \rho_d^*(r_d+p_c) & r_d \rho_d^* & r_d \rho_d^* \\
r_e \rho_e^* & r_e \rho_e^* & r_e \rho_e^* & \rho_e^*(r_e+p_d) & r_e \rho_e^*
\end{pmatrix}
\end{array}
\normalsize
\]

The eigenvalues of this matrix determine the local stability properties of the interior coexistence state.

\subsection{Characteristic Polynomial}

The characteristic equation associated with the Jacobian matrix is
$\det(\lambda I - J^*) = 0.$ For the five-species hierarchy, this yields the fifth-order characteristic polynomial

$$\lambda^5 + a_1 \lambda^4 + a_2 \lambda^3 + a_3 \lambda^2 + a_4 \lambda + a_5 = 0,$$

with coefficients

\begin{align}
a_1 &= r_a \rho_a^* + r_b \rho_b^* + r_c \rho_c^*
      + r_d \rho_d^* + r_e \rho_e^*, \\
a_2 &= - \Big(
      p_a r_a \rho_a^* \rho_b^*
      + p_b r_b \rho_b^* \rho_c^*
      + p_c r_c \rho_c^* \rho_d^* \nonumber \\
&\qquad\qquad
      + p_d r_d \rho_d^* \rho_e^*
      \Big), \\
a_3 &= p_a p_b r_a \rho_a^* \rho_b^* \rho_c^*
      + p_b p_c r_b \rho_b^* \rho_c^* \rho_d^* \nonumber \\
&\qquad\qquad
      + p_c p_d r_c \rho_c^* \rho_d^* \rho_e^*, \\
a_4 &= - \Big(
      p_a p_b p_c r_a \rho_a^* \rho_b^* \rho_c^* \rho_d^*
      \nonumber \\
&\qquad\qquad
      + p_b p_c p_d r_b \rho_b^* \rho_c^* \rho_d^* \rho_e^*
      \Big), \\
a_5 &= p_a p_b p_c p_d \,
      r_a \rho_a^* \rho_b^* \rho_c^*
      \rho_d^* \rho_e^* .
\end{align}

The signs of these coefficients provide important information regarding the stability of the equilibrium.

\subsection{Routh--Hurwitz Analysis}

Direct computation of the eigenvalues of a fifth-order polynomial is generally cumbersome. Instead, stability can be determined using the Routh--Hurwitz criterion. Define the auxiliary determinants

$\begin{aligned}
b_1 &= a_1 a_2 - a_3,\\
b_2 &= a_1 a_2 a_3 - a_1^2 a_4 - a_3^2 + a_2 a_4,\\
b_3 &= b_2 a_4 - b_1^2 a_5.
\end{aligned}$

For asymptotic stability, the Routh--Hurwitz criterion requires
$a_i > 0,$ together with $b_j > 0,$ with $i=1,2,3,4,5$ and $j=1,2,3.$

Inspection of the characteristic coefficients yields $$
a_1>0,
\qquad
a_2<0,
\qquad
a_3>0,
\qquad
a_4<0,
\qquad
a_5>0.$$

Therefore, two of the necessary Routh--Hurwitz conditions are violated for positive parameter values. Since $a_2<0$ and $a_4<0$, the Routh--Hurwitz conditions cannot be satisfied for any positive values of the demographic and predation parameters. Consequently, the characteristic polynomial necessarily possesses at least one root with a positive real part, implying that the interior coexistence equilibrium is unstable.

We therefore conclude that {\it the interior coexistence equilibrium is unstable for all} $r_i>0, p_i>0,$ and $d_i>0.$ This result has important ecological implications. Although the deterministic dynamics formally admit a state in which all five species coexist, that state cannot persist under arbitrarily small perturbations. The hierarchical interaction structure therefore prevents stable full coexistence and redirects the dynamics toward lower-dimensional coexistence states. The identification of these dynamically selected states forms the basis of the coexistence-selection rules discussed in the following section.

\section{Selection Rules for Coexistence}

The equilibrium analysis presented in the previous sections reveals a rich landscape of mathematically admissible coexistence and extinction states. However, the mere existence of an equilibrium does not imply that it is dynamically relevant or ecologically observable. In particular, the instability of the interior coexistence equilibrium demonstrates that complete coexistence cannot generally be sustained, forcing the system toward lower-dimensional coexistence states. This observation raises a more fundamental question: among the many feasible equilibrium configurations, which states are dynamically selected and what mechanisms govern their persistence?

The coexistence patterns observed in the hierarchical model should therefore be understood as the outcome of a dynamical selection process rather than as simple consequences of the interaction topology. Although the model admits numerous equilibrium configurations, stability constraints eliminate a large fraction of them from the long-term dynamics. The hierarchical interaction structure acts as a dynamical filter, progressively restricting the set of accessible biodiversity states and selecting only specific coexistence manifolds as asymptotic outcomes. From this perspective, the coexistence-selection rules identified in this study do not merely describe the topology of species interactions; rather, they characterize how hierarchical competition, indirect interactions, and demographic processes jointly constrain biodiversity and determine which subsets of species can persist over long timescales.

\subsection{Non-adjacent Coexistence States}

A distinctive feature of the hierarchical model is the privileged role played by non-adjacent species. Unlike adjacent species, which are directly coupled through predation, non-adjacent species do not interact directly and therefore avoid immediate competitive suppression. As a result, coexistence among non-adjacent species is primarily governed by demographic constraints rather than by direct predatory interactions.

This distinction is already evident from the equilibrium structure. Consecutive coexistence states generally correspond to isolated equilibrium points whose densities are uniquely determined by the balance between predation, reproduction, and mortality. In contrast, non-adjacent coexistence states frequently give rise to continuous families of equilibria characterized by conserved density sums. Examples include the two-species coexistence states $(A,C)$, $(A,D)$, $(A,E)$, $(B,D)$, $(B,E)$, and $(C,E)$, which satisfy relations given in the Eqs.\ \eqref{eq9} and \eqref{eq10}. These relations define coexistence manifolds rather than isolated equilibrium points. The surviving species may therefore redistribute their densities continuously while preserving the total occupied fraction.

The same mechanism extends to higher-order coexistence states. In particular, the three-species configuration $(A,C,E)$ occupies a special position within the hierarchy because none of the surviving species interacts directly with any other surviving species. The coexistence condition is explicitly given in Eqs.\ \eqref{eq26} and \eqref{eq27}. Consequently, the $(A,C,E)$ state forms a two-dimensional coexistence manifold and represents the highest-order coexistence configuration composed entirely of mutually non-interacting species.

These observations suggest a fundamental coexistence-selection mechanism in the hierarchical model. Species separated within the hierarchy are less susceptible to direct predatory conflicts and can therefore coexist more readily than adjacent species. The coexistence of non-adjacent species is closely related to the indirect interactions generated by the hierarchy. For example, species $A$ suppresses species $B$, while species $B$ suppresses species $C$. By reducing the abundance of $B$, species $A$ indirectly reduce the pressure exerted on $C$, thus creating an effective positive influence on $C$. This mechanism provides an intuitive explanation for the emergence of the coexistence-selection rules identified above. More generally, species that do not interact directly may become linked through such interaction chains, giving rise to indirect facilitative effects. Consequently, non-adjacent species experience weaker effective competition and are preferentially selected for long-term coexistence. The interaction hierarchy therefore acts as a filter that restricts the set of dynamically viable biodiversity states and promotes coexistence among suitably separated species.

The coexistence states identified here possess a simple graph-theoretic interpretation. Stable coexistence tends to occur among subsets of species that do not share direct competitive links. In graph-theoretic language, the surviving species form approximate independent sets of the hierarchical interaction graph. This observation suggests that coexistence selection may be understood as a structural filtering process acting on the interaction network itself. Note that, the coexistence-selection rules reported here should not be interpreted as universal ecological laws. Rather, they characterize the class of hierarchical May-Leonard systems considered in this particular study. Whether analogous selection mechanisms emerge in more general competitive networks remains an open question deserving further investigation.

\subsection{Parameter-Induced Stability Transitions}

While the interaction hierarchy determines which coexistence states are structurally favored, demographic parameters determine which of these states are dynamically realized. Small variations in reproduction or mortality rates modify the balance between growth and predation, thereby redistributing stability among competing coexistence states.

In the present model, the mortality-to-reproduction ratio acts as an effective control parameter. Variations in this ratio alter the equilibrium vacancy density and consequently modify the feasibility and stability of different coexistence configurations. As the control parameter changes, the system is expected to undergo transitions between distinct ecological outcomes, including full coexistence, partial coexistence, non-adjacent coexistence, and single-species dominance.

These considerations suggest a second coexistence-selection mechanism. Whereas the interaction hierarchy constrains the set of admissible coexistence states, demographic parameters determine which of those states become dynamically relevant. Different parameter regimes may therefore favor different subsets of species, leading to shifts in the long-term community composition.

The coexistence dynamics of the hierarchical May--Leonard model are thus governed by two complementary principles. The first is a structural selection rule arising from the interaction hierarchy itself, which favors coexistence among non-adjacent species and suppresses strongly interacting subsets. The second is a dynamical selection rule arising from parameter variation, which controls the redistribution of stability among the admissible coexistence states.

These analytical predictions provide a framework for interpreting the numerical results presented in the following section, where the parameter-induced transitions between coexistence states are investigated explicitly.

\begin{figure*}[!t]
\centering
\includegraphics[width=0.65\textwidth]{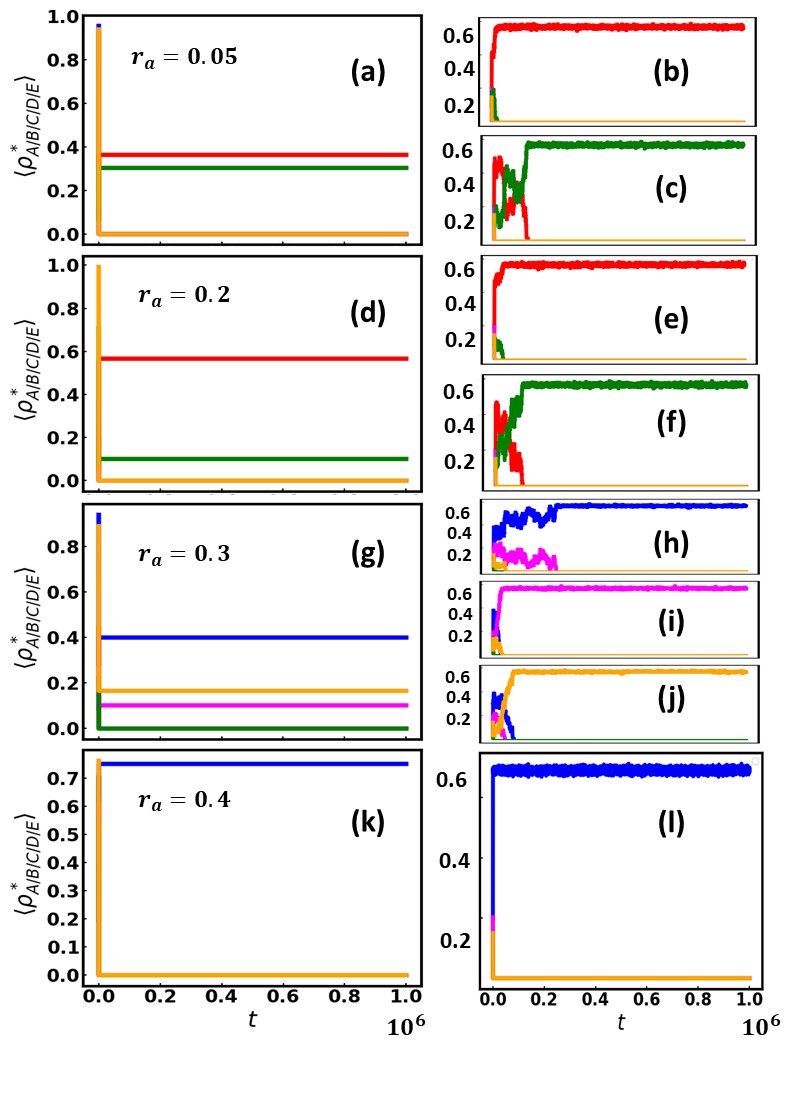}
\caption{%
Time evolution of species densities for different values of the reproduction rate parameter $r_a$. The left column shows the numerical solutions of the ODE model, while the right column presents results obtained from Monte Carlo (MC) simulations. Panels~(a)--(f) correspond to $r_a=0.05$ and $r_a=0.2$, where species B (red) and D (green) survive in the long-time limit. Panels~(g)--(j), corresponding to $r_a=0.3$, exhibit asymptotic coexistence of species A (blue), C (magenta), and E (orange). Panels~(k)--(l) show that species A (blue) is the sole survivor. For $r_a=0.05$, $0.2$, and $0.3$, the ODE model predicts the long-time persistence of multiple species, whereas each individual MC realization ultimately converges to a single surviving species due to stochastic fluctuations. The remaining parameters are fixed at $r_b=r_c=r_d=r_e=0.3$, $d_a=d_b=d_c=d_d=d_e=0.1$, and $p_a=p_b=p_c=p_d=0.2$.%
}
\label{fig:1}
\end{figure*}

\section{Graph-Theoretic Interpretation and N-Species Extension}


The coexistence-selection rules identified above suggest a natural graph-theoretic interpretation of the hierarchical model. The interaction structure

\[
A_1 \rightarrow A_2 \rightarrow A_3 \rightarrow \cdots \rightarrow A_N
\]

can be represented as a path graph, where each vertex corresponds to a species and each edge represents a direct hierarchical interaction. The analytical results (later validated through numerical experiments) obtained for the five-species system indicate that stable coexistence occurs preferentially among species that are not directly connected in this graph. In graph-theoretic language, the surviving species tend to form independent sets of the interaction graph.

This observation leads to a simple general result for hierarchical systems containing an arbitrary number of species, assuming that long-term coexistence is possible only among species that are not directly connected by a predatory interaction.

\begin{theorem}[Maximum possible coexistence in an $N$-species hierarchy]
Consider an $N$-species hierarchical May--Leonard system whose interaction network forms a linear hierarchy

\[
A_1 \rightarrow A_2 \rightarrow A_3 \rightarrow \cdots \rightarrow A_N .
\]

If long-term coexistence is restricted to species subsets that contain no directly interacting neighbors, then the maximum number of species that can simultaneously coexist is

\[
N_{\max}
=
\left\lceil \frac{N}{2}\right\rceil .
\]
\end{theorem}

\begin{proof}
Let $k \in \mathbb{Z}^{+}$ denote the number of species belonging to a coexistence subset. Since directly connected species cannot simultaneously belong to the same coexistence subset, every pair of selected species must be separated by at least one unselected species.

Consequently, a coexistence subset containing $k$ species requires at least $k-1$ additional species acting as separators. The minimum number of species needed to accommodate such a subset is therefore

\[
k+(k-1)=2k-1.
\]

Since the hierarchy contains only $N$ species,

\[
2k-1 \le N.
\]

Rearranging yields

\[
k \le \frac{N+1}{2}.
\]

Because $k$ must be an integer,

\[
k \le \left\lceil \frac{N}{2}\right\rceil .
\]

Thus no coexistence subset can contain more than
\[
\left\lceil \frac{N}{2}\right\rceil
\]
species.

\begin{figure*}[t]
\centering
\includegraphics[width=\textwidth]{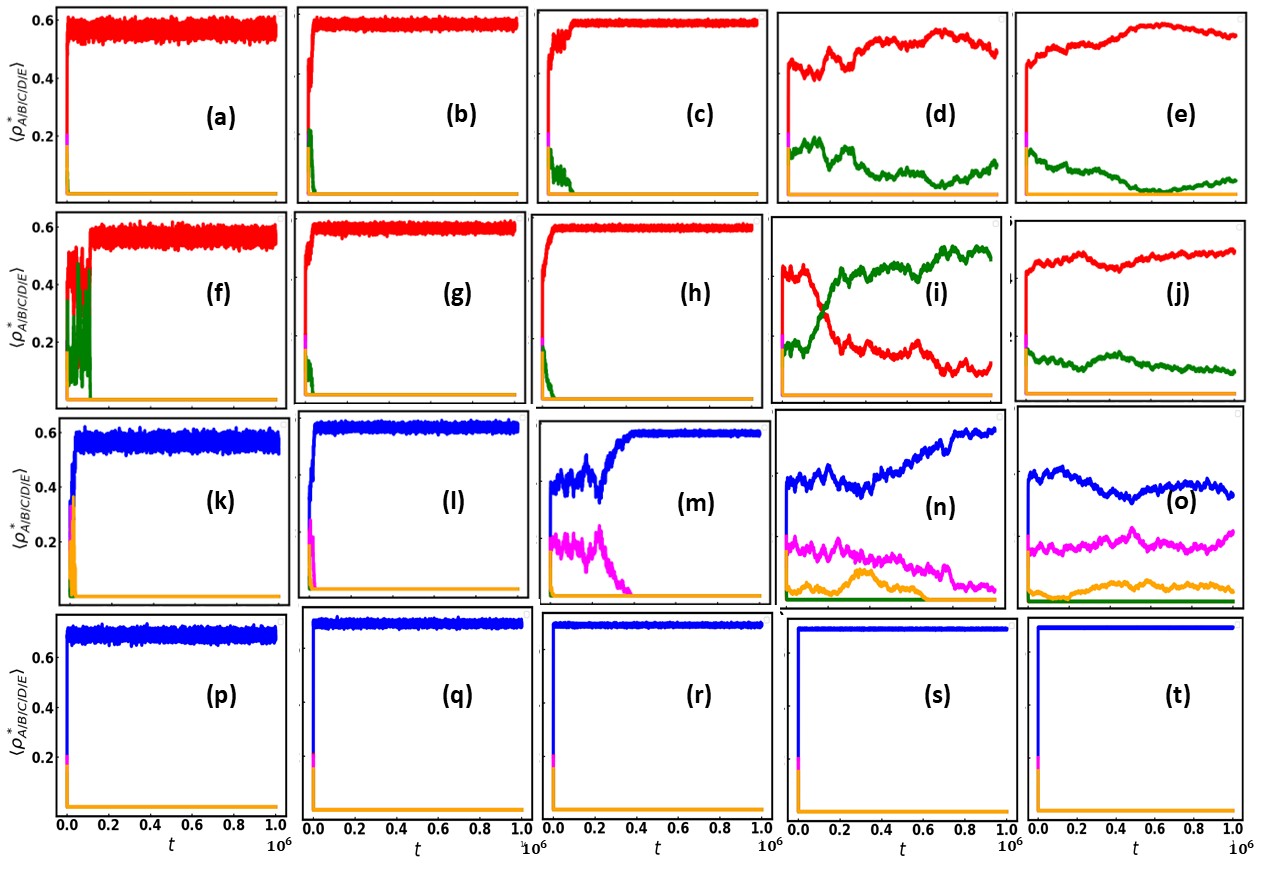}
\caption{
Species size strongly influences coexistence behavior in the stochastic system. We present the time evolution of species densities in the asymptotic time limit for different values of \(r_a\) and system sizes \(L\), where each row corresponds to a fixed \(r_a\) and each column to a fixed system size. For \(r_a = 0.1\) [panels (a)–(e)], only species \(B\) survives for smaller system sizes \(L = 50\), \(100\), and \(200\), whereas coexistence between species \(B\) and \(D\) emerges for larger sizes \(L = 500\) and \(800\). A similar qualitative behavior is observed for \(r_a = 0.2\) [panels (f)–(j)], where coexistence appears only when the system size becomes sufficiently large. For \(r_a = 0.3\) [panels (k)–(o)], species \(A\) solely survives for \(L = 50\), \(100\), and \(200\); however, for \(L = 500\), species \(A\) and \(C\) maintain nonzero densities while species \(E\) survives transiently before eventually going extinct, and for \(L = 800\), stable coexistence among species \(A\), \(C\), and \(E\) is observed. In contrast, for larger values of \(r_a\) [panels (p)–(t)], the dynamics approach a single-species absorbing state where only species \(A\) survives irrespective of the system size.  
}
\label{fig:6}
\end{figure*}

To show that this upper bound is attainable, consider the subset formed by selecting every alternate species,

\[
\{A_1,A_3,A_5,\ldots\}.
\]

No two species in this set are adjacent in the hierarchy, and therefore the subset contains exactly

\[
\left\lceil \frac{N}{2}\right\rceil
\]

species. Hence the upper bound is achieved, proving that

\[
N_{\max}
=
\left\lceil \frac{N}{2}\right\rceil .
\]

\end{proof}

For the five-species hierarchy investigated earlier, we have

\[
N=5,
\]

and therefore

\[
N_{\max}
=
\left\lceil \frac{5}{2}\right\rceil
=
3.
\]

The coexistence state $(A,C,E)$ observed in the stability analysis in Appendix \ref{Five-species} therefore realizes the largest coexistence subset permitted by the hierarchical interaction structure. More generally, the theorem suggests that coexistence in hierarchical systems is constrained by simple structural principles that become increasingly restrictive as the number of species increases. The hierarchy thus acts as a dynamical filter, selecting only a small subset of all theoretically possible biodiversity states.

Note that we explicitly assume to prove the previous theorem that the coexistence states identified in the five-species hierarchy suggest a general principle. In particular, all stable coexistence states observed in our model involve species that are not directly connected through the hierarchical interaction chain. This property can be established analytically for the simplified system with identical reproduction and mortality rates, \(r_i=r\) and \(d_i=d\).

\begin{proposition}
Consider the $N$-species hierarchical May--Leonard model

\[
A_1 \rightarrow A_2 \rightarrow A_3 \rightarrow \cdots \rightarrow A_N,
\]

with \(r_i=r\), \(d_i=d\), and \(p_i>0\) for all \(i\). Then no equilibrium can contain two adjacent species \(A_i\) and \(A_{i+1}\) simultaneously with positive densities.
\end{proposition}

\begin{proof}
We prove the result by induction along the hierarchy. Let, $S=\sum\limits_i^N \rho_i$.

\medskip

\noindent
\textbf{Base case:} Consider the first adjacent pair \((A_1,A_2)\). Suppose that

\[
\rho_1^*>0,
\qquad
\rho_2^*>0.
\]

Since both densities are positive, their equilibrium conditions imply

\[
r(1-S)-d=0,
\]

and

\[
r(1-S)-p_1\rho_1^*-d=0.
\]

Subtracting the two equations yields

\[
p_1\rho_1^*=0.
\]

Since \(p_1>0\), it follows that

\[
\rho_1^*=0,
\]

which contradicts the assumption \(\rho_1^*>0\). Therefore species \(A_1\) and \(A_2\) cannot coexist.

\medskip

\noindent
\textbf{Induction hypothesis:} Assume that for some \(i\ge 2\), the adjacent pair \((A_{i-1},A_i)\) cannot coexist. Equivalently,

\[
\rho_{i-1}^*>0
\quad\Longrightarrow\quad
\rho_i^*=0.
\]

\medskip

\noindent
\textbf{Induction step:} We show that the pair \((A_i,A_{i+1})\) also cannot coexist.

Suppose, on the contrary, that

\[
\rho_i^*>0,
\qquad
\rho_{i+1}^*>0.
\]

Since both species are present at equilibrium, their stationarity conditions are

\[
r(1-S)-p_{i-1}\rho_{i-1}^*-d=0,
\]

and

\[
r(1-S)-p_i\rho_i^*-d=0.
\]

Subtracting these two equations gives

\[
p_{i-1}\rho_{i-1}^*
=
p_i\rho_i^*.
\]

Because \(p_i>0\) and \(\rho_i^*>0\), the right-hand side is strictly positive. Therefore,

\[
p_{i-1}\rho_{i-1}^*>0,
\]

which implies

\[
\rho_{i-1}^*>0.
\]

Hence we obtain

\[
\rho_{i-1}^*>0,
\qquad
\rho_i^*>0.
\]

However, this contradicts the induction hypothesis that species \(A_{i-1}\) and \(A_i\) cannot coexist.

Therefore the assumption

\[
\rho_i^*>0,
\qquad
\rho_{i+1}^*>0
\]

must be false. Consequently, species \(A_i\) and \(A_{i+1}\) cannot coexist.

Thus, by mathematical induction, no equilibrium can contain two adjacent species simultaneously with positive densities.
\end{proof}

The proposition implies that every coexistence equilibrium must be composed exclusively of mutually nonadjacent species in the symmetric framework with identical reproduction and mortality rates, $r_i=r$ and $d_i=d$ for all species. In graph-theoretic language, the surviving species form an independent set of the hierarchical interaction graph and by the previous theorem, for an $N$-species hierarchy, the largest coexistence subset contains

\[
N_{\max}
=
\left\lceil \frac{N}{2}\right\rceil
\]

species.

\subsection{Enumeration of admissible coexistence states}

Here, based on the analytical results obtained for the five-species system, we assume that long-term coexistence is restricted to subsets of mutually nonadjacent species. In graph-theoretic terms, the surviving species are assumed to form independent sets of the hierarchy graph. This assumption allows us to derive general coexistence bounds and counting results for hierarchical systems containing an arbitrary number of species. Under the independent-set assumption, the number of admissible coexistence states can be enumerated exactly.

Let \(I_N\) denote the number of independent sets of the path graph \(P_N\).

For \(N=1\), the independent sets are

\[
\varnothing,\qquad \{A_1\},
\]

so that

\[
I_1=2.
\]

For \(N=2\), the independent sets are

\[
\varnothing,\qquad \{A_1\},\qquad \{A_2\},
\]

and therefore

\[
I_2=3.
\]

For a path graph containing \(N\) vertices, every independent set belongs to one of two classes:

\begin{enumerate}
\item \(A_N\) is not selected. Then the remaining vertices form an arbitrary independent set of \(P_{N-1}\), giving \(I_{N-1}\) possibilities.

\item \(A_N\) is selected. Then \(A_{N-1}\) cannot be selected, and the remaining vertices form an arbitrary independent set of \(P_{N-2}\), giving \(I_{N-2}\) possibilities.
\end{enumerate}

Hence

\[
I_N=I_{N-1}+I_{N-2}.
\]

Together with

\[
I_1=2,
\qquad
I_2=3,
\]

this recurrence generates

\[
2,3,5,8,13,\ldots,
\]

which is precisely the Fibonacci sequence shifted by two indices. Therefore

\[
I_N=F_{N+2}.
\]

Consequently, the number of admissible coexistence states grows only according to the Fibonacci sequence, whereas the total number of possible species subsets grows as $2^N.$

Therefore

\[
\frac{I_N}{2^N}
=
\frac{F_{N+2}}{2^N}.
\]

Using the asymptotic relation
\[
F_n\sim \frac{\phi^n}{\sqrt5},
\qquad
\phi=\frac{1+\sqrt5}{2},
\]
we obtain
\[
\frac{F_{N+2}}{2^N}
\sim
\frac{\phi^2}{\sqrt5}
\left(\frac{\phi}{2}\right)^N.
\]
Since \(\phi/2<1\), the ratio tends to zero as \(N\to\infty\), i.e.,
\[
\lim_{N\to\infty}
\frac{F_{N+2}}{2^N}
=
0.
\]

Therefore, increasing community size does not necessarily increase the number of viable coexistence states. Instead, the hierarchical interaction structure eliminates an increasingly large proportion of potential species assemblages, thereby restricting long-term coexistence to a relatively small subset of all possible biodiversity states.

\begin{figure}[t]
\centering
\includegraphics[width=\columnwidth]{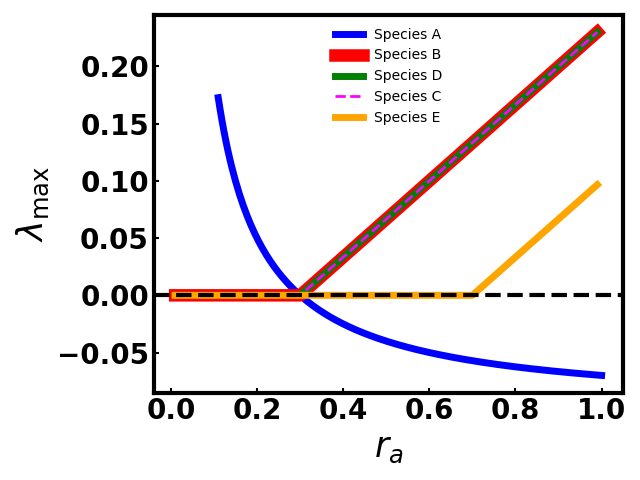}
\caption{
Stability analysis of single-species equilibria is performed by numerically plotting the largest eigenvalue of the Jacobian matrix associated with single-species equilibria. From the plot, it is evident that up to $r_a = 0.3$, species A remains unstable while the other species are neutrally stable along this line. Beyond $r_a = 0.3$, species B, C, and D enter an unstable regime, whereas species E continues to remain neutrally stable. In contrast, species A transitions to a stable regime after $r_a = 0.3$, as indicated by its largest eigenvalue becoming negative. The remaining parameters are set as $r_b = r_c = r_d = r_e = 0.3$, $d_a = d_b = d_c = d_d = d_e = 0.1$, and $p_a = p_b = p_c = p_d = 0.2$.
}
\label{fig:2}
\end{figure}

\section{Numerical Simulations} \label{numerical}

To complement the analytical results and investigate the long-time behavior of the hierarchical May-Leonard model, we perform both deterministic and stochastic simulations. Throughout this section, the parameters are fixed at
$r_b=r_c=r_d=r_e=0.3,
d_a=d_b=d_c=d_d=d_e=0.1,
p_a=p_b=p_c=p_d=0.2,$ while the reproduction rate of species \(A\), \(r_a\), is varied systematically. The numerical simulations are used to test the coexistence-selection rules derived in the previous sections and to determine how demographic fluctuations modify the deterministic predictions. It should be emphasized that the transition sequence reported here is obtained by varying $r_a$ while keeping the remaining parameters fixed. The purpose is not to identify universal critical values, but rather to illustrate the generic coexistence-selection mechanism generated by hierarchical interactions.

\subsection{Mean-Field and Stochastic Dynamics}

Figure~\ref{fig:1} compares the numerical solutions of the mean-field equations with Monte Carlo simulations for different values of \(r_a\). For small values of \(r_a\), the deterministic dynamics converge to a coexistence state involving species \(B\) and \(D\) (see left column, Fig.~\ref{fig:1}). Increasing \(r_a\) to approximately \(0.3\) produces a transition to the non-adjacent coexistence state \((A,C,E)\), while further increasing \(r_a\) results in complete dominance of species \(A\). These results are consistent with the analytical prediction that only selected subsets of species remain stable under hierarchical competition. The linear stability analysis of the \((A,C,E)\) coexistence manifold for $r_i=r$ and $d_i=d$ shows that all transverse eigenvalues are non-positive, while the remaining zero eigenvalues correspond to neutral directions along the coexistence manifold. Details are provided in Appendix~\ref{Five-species}.

{ We validate these results by numerically integrating Eqs.\ \eqref{diff_a} and~\eqref{diff_all} and comparing their asymptotic behavior with analytical predictions (Fig.~\ref{fig:1}). In parallel, stochastic Monte Carlo (MC) simulations are performed on a two-dimensional lattice showing overall agreement with the ODE dynamics (see Appendix~\ref{monte-carlo} for the MC scheme). Figure~\ref{fig:1} illustrates the time evolution for different $r_a$. For $r_a = 0.05$ (Fig ~\ref{fig:1} a) and $0.2$ (Fig.~\ref{fig:1} d), the ODE model predicts coexistence of species B and D, whereas MC simulations yield single-species survival per realization (either B or D) (Fig.~\ref{fig:1} b,c and e,f). For $r_a = 0.3$, (Fig.~\ref{fig:1} g) the ODE model exhibits coexistence of species A, C, and E, while MC simulations (Fig ~\ref{fig:1} h,i) again result in dominance of a single species in each realization. At $r_a = 0.4$, (Fig.~\ref{fig:1} k and l) both approaches show consistent behavior. Notably, for lower $r_a$, the total density in the ODE model closely matches the summed densities from MC realizations, with deviations attributable to the mean-field approximation in the ODE framework versus nearest-neighbor interactions in the stochastic model. }

\begin{figure}[t]
\centering
\includegraphics[width=\columnwidth]{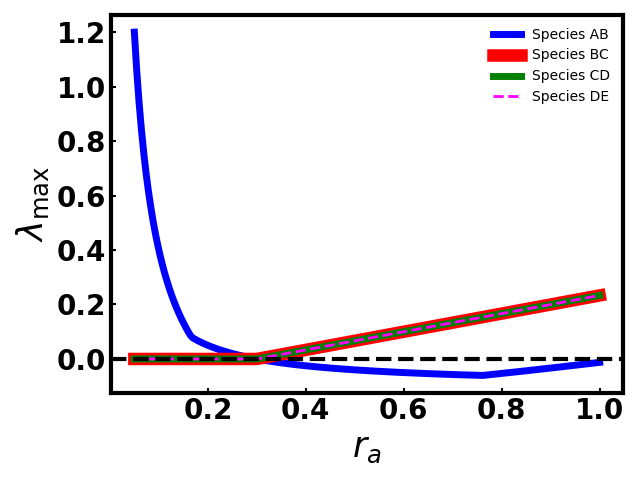}
\caption{
The largest eigenvalue of the Jacobian matrix associated with two-species coexistence equilibria is numerically plotted in order to perform the stability analysis of two-species coexistence equilibria. The resulting plot shows that species AB remains unstable up to $r_a = 0.3$, while the other species exhibit neutral stability along this line. However, species BC, CD, and DE go into an unstable phase when $r_a$ is greater than 0.3. In contrast, species AB's largest eigenvalue becomes negative after $r_a = 0.3$, indicating a change to a stable state; nonetheless, the fixed point is not feasible in this region. $r_b = r_c = r_d = r_e = 0.3$, $d_a = d_b = d_c = d_d = d_e = 0.1$, and $p_a = p_b = p_c = p_d = 0.2$ are the values of the other parameters.
}
\label{fig:3}
\end{figure}

\subsection{Finite-Size Effects}

The role of demographic fluctuations is further examined in Fig.\ \ref{fig:6}, which presents Monte Carlo simulations for different system sizes at different values of \(r_a\). For small populations (i.e. small system-size), demographic fluctuations dominate the dynamics and rapidly eliminate coexistence. However, increasing the system-size significantly enhances the persistence of multiple species. Figure~\ref{fig:6} illustrates that for relatively small populations, the dynamics are dominated by single-species survival. Specifically, for system sizes \(L = 50\), \(100\), and \(200\) (shown in the first three columns), only one species survives within the considered time window, irrespective of the value of \(r_a\).
However, as the population size is further increased, a qualitative change in the dynamics emerges. In the fourth and fifth columns, corresponding to \(L = 500\) and \(800\), coexistence among species is observed for \(r_a = 0.1\), \(0.2\), and \(0.3\). These results indicate that increasing the system size substantially enhances the probability of long-term species coexistence in the stochastic system. Moreover, the corresponding ODE-based mean-field model also predicts coexistence of multiple species for \(r_a = 0.1\), \(0.2\), and \(0.3\). 
In particular, the coexistence states predicted by the mean-field equations become increasingly visible in the MC results with growing population size. These results indicate that the discrepancy between deterministic and stochastic dynamics can be attributed to the finite-size effects mainly.

\subsection{Stability of Single-Species Equilibria}

To understand the emergence of single-species dominance, we compute the largest Jacobian eigenvalue associated with each single-species equilibrium. The results are shown in Fig.~\ref{fig:2}. A clear stability transition occurs near \(r_a = 0.3\). For smaller values of \(r_a\), the species-\(A\) equilibrium is unstable. As \(r_a\) increases, its dominant eigenvalue becomes negative, indicating stabilization of the species-\(A\) state. Simultaneously, the remaining single-species equilibria become unstable. This transition explains the emergence of species-\(A\) dominance observed in Fig.~\ref{fig:1} for sufficiently large values of \(r_a\). The numerical observations presented here are consistent with the corresponding linear stability analysis, which is provided in detail in Appendix~\ref{Five-species}.

\subsection{Stability of Two-Species Coexistence States}

The largest eigenvalues associated with representative two-species coexistence states are shown in Fig.~\ref{fig:3}. The coexistence state involving species \(A\) and \(B\) undergoes a stability transition as \(r_a\) varies, whereas the coexistence states \((B,C)\), \((C,D)\), and \((D,E)\) remain unstable throughout most of the parameter range. These results support the analytical prediction that direct hierarchical interactions tend to destabilize coexistence. Consequently, only a restricted subset of coexistence states remains dynamically relevant.

\subsection{Probability of Long-Time Ecological Outcomes}


\begin{figure}[t]
\centering
\includegraphics[width=\columnwidth]{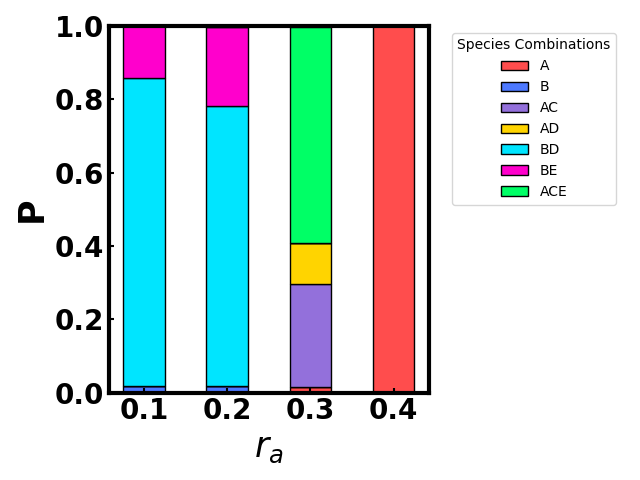}
\caption{
The survival probabilities of different species were analyzed under varying initial conditions, with four distinct values of the reproduction rate of species A ($r_a$). For $r_a = 0.1$, species B and D consistently coexist. When $r_a = 0.2$, the system exhibits bistability, where either species B and D coexist, or species B and E coexist, depending on the initial conditions. At $r_a = 0.3$, the dynamics become more diverse: in some cases, species A, C, and E coexist, while in others, only species A and C persist. For $r_a = 0.4$ and higher, the system stabilizes to a single-species outcome, with only species A surviving across all initial conditions. In this simulation, we randomly change the initial conditions for all five species, keeping the total density of all species and the vacant space constant.}
\label{fig:4}
\end{figure}

\begin{figure*}[t]
\centering
\includegraphics[width=\textwidth]{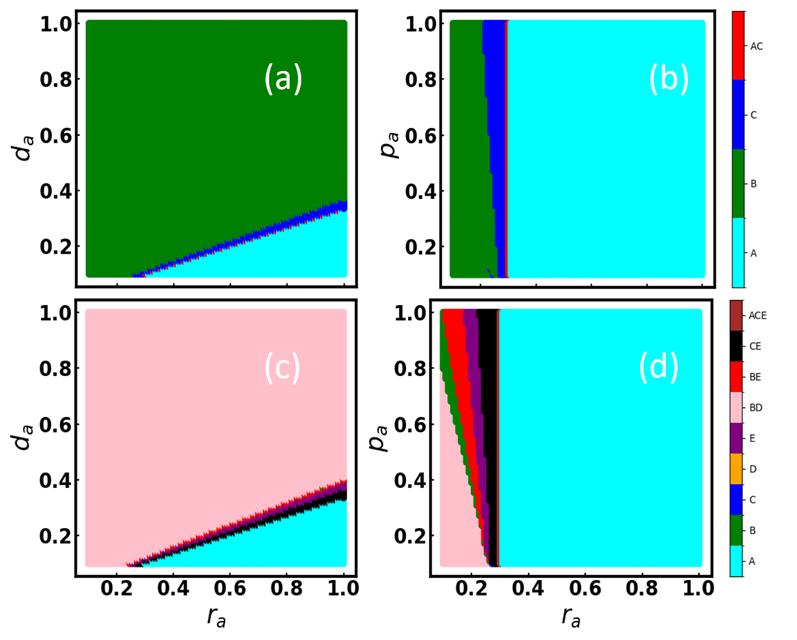}
\caption{
Phase-space diagrams showing the asymptotic (long-time) dynamics of multi-species systems under systematic variation of key parameters. (a) For the three-species model given in Appendix \ref{Three-species}, the death rate $d_a$ and reproduction rate $r_a$ of species A are varied, and the governing equations, Eqs.\ \eqref{diff_a} and \eqref{diff_all}, are numerically integrated to the steady state. (b) Same as (a), but with the predation rate $p_a$ and reproduction rate $r_a$ of species A as control parameters. (c) and (d) Corresponding phase-space plots for the five-species system, using the same parameter variations as in (a) and (b), respectively, illustrating the effect of increased system dimensionality on the emergent dynamics. All remaining parameters are fixed at $r_b = r_c = r_d = r_e = 0.3$, $d_b = d_c = d_d = d_e = 0.1$, and $p_b = p_c = p_d = 0.2$. For panels (a) and (c), $p_a = 0.2$, while for panels (b) and (d), $d_a = 0.1$.  
}
\label{fig:5}
\end{figure*}

To quantify the stochastic selection of coexistence states, we perform many independent realizations and recorded the final ecological outcome. The corresponding probabilities are shown in Fig.~\ref{fig:4}. 
For \(r_a=0.1\) and \(0.2\), the coexistence state \((B, D)\) dominates the probability distribution. Near \(r_a=0.3\), several competing coexistence states become accessible, including \((A,C)\), \((A,D)\), and \((A,C,E)\). This region corresponds to a redistribution of stability among competing attractors. For \(r_a=0.4\), the probability becomes concentrated entirely on the species-\(A\) state, indicating complete dominance. The probability distributions therefore provide direct evidence for parameter-induced stability transitions and confirm that changes in \(r_a\) systematically alter the selected biodiversity state.

\subsection{Global Phase Structure}

The overall organization of coexistence states in parameter space is illustrated in Fig.~\ref{fig:5}. The phase diagrams reveal several distinct ecological regions separated by well-defined transition boundaries. Large portions of parameter space are occupied by single-species dominance states, whereas coexistence states occur only within specific regions. Among the coexistence states, non-adjacent combinations such as \((A,C,E)\) occupy comparatively larger regions than coexistence states involving directly interacting species. These results provide a global visualization of the coexistence-selection rules identified analytically. Rather than producing arbitrary biodiversity patterns, the hierarchical interaction structure restricts the set of admissible coexistence states and favors species combinations that minimize direct predatory conflicts.

Overall, the numerical simulations strongly support the analytical framework developed in this study. The results demonstrate that coexistence in the hierarchical May-Leonard model is governed by two complementary mechanisms: a structural selection rule imposed by the interaction hierarchy and a dynamical selection rule controlled by demographic parameters. Together, these mechanisms determine which species subsets persist in the long-term limit and how biodiversity is reorganized as ecological conditions vary.

\section{Conclusion}

In this work, we have investigated a hierarchical extension of the May-Leonard model in which species interact through a directed predation chain rather than a cyclic interaction network. By incorporating reproduction, natural mortality, and hierarchical predation within a mean-field framework, we have characterized the equilibrium structure of a five-species competitive system and have identified the coexistence states permitted by the hierarchy.

The analysis has revealed a rich spectrum of ecological outcomes, including extinction, single-species survival, partial coexistence, and full coexistence. A central result of the study is that the interior coexistence state, in which all species survive simultaneously, is generically unstable for positive demographic and predation parameters. Consequently, the long-time dynamics are redirected toward lower-dimensional coexistence states involving only selected subsets of species.

The equilibrium structure further reveals that coexistence is not arbitrary. Non-adjacent species can coexist through continuous coexistence manifolds, whereas directly interacting species are subject to stronger competitive constraints. This leads to a coexistence-selection mechanism in which the interaction hierarchy itself restricts the set of dynamically viable biodiversity states. In particular, coexistence among non-adjacent species emerges as a recurring and robust feature of the hierarchical system.

Numerical simulations support the analytical predictions and demonstrate that demographic parameters redistribute stability among competing coexistence states. As the control parameter $r_a$ is varied, the system undergoes transitions between different coexistence configurations and single-species dominance states. Stochastic simulations further show that demographic fluctuations can eliminate coexistence in finite populations, although increasing system size progressively restores the coexistence patterns predicted by the deterministic mean-field description.

Overall, the present work extends the May-Leonard framework from cyclic competition to hierarchical interaction networks and demonstrates that biodiversity patterns are governed by a combination of interaction topology and demographic parameters. The resulting coexistence-selection rules provide a theoretical basis for understanding how hierarchical competitive structures organize species persistence and community composition in complex ecological systems.

The framework developed here opens several directions for future research, including the study of more general hierarchical networks, environmental heterogeneity, and stochastic spatial dynamics. Extending the present analysis beyond the symmetric parameter regime and toward arbitrary interaction topologies may ultimately contribute to a more general theory of coexistence and the organization of biodiversity in hierarchical ecological systems. Natural extensions include incorporating spatial structure to investigate pattern formation and to reconcile deterministic mean-field predictions with stochastic dynamics; generalizing the analysis to arbitrary $N$-species hierarchies and more general directed interaction networks; and accounting for environmental heterogeneity, temporal variability, and adaptive ecological interactions. Complementing Monte Carlo simulations with analytical stochastic approaches may further elucidate the roles of demographic fluctuations and extinction phenomena. Together, these developments would broaden the applicability of the present framework and contribute to a more general theoretical understanding of how hierarchical interactions govern species coexistence, biodiversity organization, and ecological stability in complex ecological communities.

{\color{red}



}

{\color{black}

\appendix
\section{Three-Species Hierarchical Model}
\label{Three-species}

To provide analytical insight into the coexistence-selection mechanisms of the five-species hierarchical model, we consider here a reduced system consisting of three species arranged in a similar hierarchical fashion. Despite its simplicity, this reduced model captures the essential mechanisms responsible for coexistence, exclusion, and stability transitions observed in the full system.

\subsection{Model Equations}

Let $\rho_a$, $\rho_b$, and $\rho_c$ denote the densities of species $A$, $B$, and $C$, respectively, and let $\rho_v = 1-\rho_a-\rho_b-\rho_c$ denotes the density of vacant sites. The mean-field equations are

\begin{align}
\dot{\rho}_a
&=
\rho_a
\left(
r_a \rho_v - d_a
\right),
\\
\dot{\rho}_b
&=
\rho_b
\left(
r_b \rho_v
-
p_a \rho_a
-
d_b
\right),
\\
\dot{\rho}_c
&=
\rho_c
\left(
r_c \rho_v
-
p_b \rho_b
-
d_c
\right).
\end{align}

Species $A$ suppresses species $B$, species $B$ suppresses species $C$, and no direct interaction exists between species $A$ and $C$.

\subsection{Equilibrium States}

Setting $\dot{\rho}_a=\dot{\rho}_b=\dot{\rho}_c=0$ yields the following equilibrium configurations:

\begin{itemize}

\item Extinction equilibrium:
\[
E_0=(0,0,0).
\]

\item Single-species equilibria:
\begin{align*}
E_A &= \left(1-\frac{d_a}{r_a},0,0\right),\\
E_B &= \left(0,1-\frac{d_b}{r_b},0\right),\\
E_C &= \left(0,0,1-\frac{d_c}{r_c}\right).
\end{align*}

\item Two-species coexistence equilibrium:
\[
E_{AC}=(\rho_a^*,0,\rho_c^*),
\qquad
\rho_a^*+\rho_c^*
=
1-\frac{d_a}{r_a},
\]
provided
\[
\frac{d_a}{r_a}
=
\frac{d_c}{r_c}.
\]

\item Interior coexistence equilibrium:
\[
E_{ABC}
=
(\rho_a^*,\rho_b^*,\rho_c^*),
\]
with
\[
\rho_a^*
=
\frac{r_b(d_a/r_a)-d_b}{p_a},
\qquad
\rho_b^*
=
\frac{r_c(d_a/r_a)-d_c}{p_b},
\]
and
\[
\rho_c^*
=
1-\frac{d_a}{r_a}
-\rho_a^*
-\rho_b^*.
\]

\end{itemize}

\subsection{Global Stability of the Extinction Equilibrium}

In the following, we derive a sufficient condition for global extinction of the three-species model. We show that if the mortality rate of every species exceeds its reproduction rate, then all populations ultimately vanish and the extinction equilibrium is globally asymptotically stable.

\begin{theorem}
If
\[
r_a<d_a,
\qquad
r_b<d_b,
\qquad
r_c<d_c,
\]
then the extinction equilibrium
\[
E_0=(0,0,0)
\]
is globally asymptotically stable in the nonnegative orthant
\[
\mathbb{R}_{\ge0}^{3}.
\]
\end{theorem}

\begin{proof}
Consider the Lyapunov function
\[
V(\rho_a,\rho_b,\rho_c)
=
\rho_a+\rho_b+\rho_c
=
S.
\]

Since \(V\ge0\) and \(V=0\) only at the origin, \(V\) is positive definite on
\(\mathbb{R}_{\ge0}^{3}\).

Differentiating along trajectories yields
\[
\dot V
=
\dot\rho_a+\dot\rho_b+\dot\rho_c.
\]

Using the governing equations,
\[
\dot V
=
\rho_a(r_a\rho_v-d_a)
+
\rho_b(r_b\rho_v-p_a\rho_a-d_b)
+
\rho_c(r_c\rho_v-p_b\rho_b-d_c).
\]

Since \(\rho_v\le1\) and all predation terms are nonnegative,
\[
\dot V
\le
\rho_a(r_a-d_a)
+
\rho_b(r_b-d_b)
+
\rho_c(r_c-d_c).
\]

Under the conditions
\[
r_a<d_a,
\qquad
r_b<d_b,
\qquad
r_c<d_c,
\]
it follows that
\[
\dot V<0
\]
for every nonzero state in \(\mathbb{R}_{\ge0}^{3}\).

Hence \(V\) is a strict Lyapunov function. Since the nonnegative orthant is forward invariant, every trajectory converges to the origin. Therefore the extinction equilibrium \(E_0\) is globally asymptotically stable.
\end{proof}

\subsection{Local Stability of Single-Species Equilibria}

The three-species model admits three single-species equilibria, which exist whenever $r_i>d_i$. Linearization about each equilibrium yields the following eigenvalue spectra:

\[
\begin{aligned}
\mathrm{Spec}(E_A)
&=
\left\{
-r_a\rho_a^*,
r_b(1-\rho_a^*)-p_a\rho_a^*-d_b,
r_c(1-\rho_a^*)-d_c
\right\},
\\[2mm]
\mathrm{Spec}(E_B)
&=
\left\{
r_a-d_a,
-r_b\rho_b^*,
r_c(1-\rho_b^*)-p_b\rho_b^*-d_c
\right\},
\\[2mm]
\mathrm{Spec}(E_C)
&=
\left\{
r_a-d_a,
r_b-d_b,
-r_c\rho_c^*
\right\}.
\end{aligned}
\]

Since one eigenvalue is always negative due to self-limitation, stability is determined by the signs of the remaining invasion eigenvalues. Consequently,

\begin{itemize}

\item $E_A$ is locally asymptotically stable if
\[
r_b(1-\rho_a^*)-p_a\rho_a^*-d_b<0,
\quad
r_c(1-\rho_a^*)-d_c<0.
\]

\item $E_B$ is locally asymptotically stable if
\[
r_a-d_a<0,
\quad
r_c(1-\rho_b^*)-p_b\rho_b^*-d_c<0.
\]

\item $E_C$ is locally asymptotically stable if
\[
r_a-d_a<0,
\quad
r_b-d_b<0.
\]

\end{itemize}

These conditions show that local stability of a single-species equilibrium is governed by the inability of the absent species to successfully invade the resident population.

\subsection{Local Stability of the Two-Species Coexistence State}

The three-species model admits a coexistence equilibrium involving the nonadjacent species $A$ and $C$,

\[
E_{AC}=(\rho_a^*,0,\rho_c^*),
\]

with

\[
\rho_a^*+\rho_c^*
=
1-\frac{d_a}{r_a},
\]

provided that

\[
\frac{d_a}{r_a}
=
\frac{d_c}{r_c}.
\]

The Jacobian matrix evaluated at $E_{AC}$ is

\[
J(E_{AC})=
\begin{pmatrix}
-r_a\rho_a^* & -r_a\rho_a^* & -r_a\rho_a^*\\
0 & r_b\rho_a^*-d_b & 0\\
-r_c\rho_c^* & 0 & -r_c\rho_c^*
\end{pmatrix}.
\]

The corresponding eigenvalues are

\[
\lambda_1=0,
\qquad
\lambda_2=-(r_a\rho_a^*+r_c\rho_c^*),
\qquad
\lambda_3=r_b\rho_a^*-d_b.
\]

The zero eigenvalue reflects the continuum of equilibria satisfying

\[
\rho_a^*+\rho_c^*
=
1-\frac{d_a}{r_a},
\]

and therefore corresponds to perturbations along the coexistence manifold. Since

\[
\lambda_2<0,
\]

the coexistence state is stable in one transverse direction. The remaining stability condition is determined by the invasion eigenvalue

\[
\lambda_3=r_b\rho_a^*-d_b.
\]

Hence $E_{AC}$ is locally stable whenever

\[
r_b\rho_a^*-d_b<0,
\]

and unstable otherwise. Consequently, the coexistence equilibrium is non-hyperbolic, neutrally stable along the coexistence manifold, and locally stable to transverse perturbations when the absent species $B$ cannot invade. This result demonstrates that coexistence naturally emerges between nonadjacent species in the hierarchy. Since species $A$ and $C$ do not interact directly, the inhibitory chain $A\rightarrow B\rightarrow C$ weakens effective competition and permits the formation of a stable coexistence manifold.

\subsection{Instability of the Interior Coexistence Equilibrium}

The Jacobian evaluated at $E_{ABC}$ is

\[
J^*
=
-
\begin{pmatrix}
r_a\rho_a^* & r_a\rho_a^* & r_a\rho_a^*\\
(r_b+p_a)\rho_b^* & r_b\rho_b^* & r_b\rho_b^*\\
r_c\rho_c^* & (r_c+p_b)\rho_c^* & r_c\rho_c^*
\end{pmatrix}.
\]

The characteristic polynomial is

\[
\lambda^3+a_1\lambda^2+a_2\lambda+a_3=0,
\]

with

\[
a_1
=
r_a\rho_a^*
+
r_b\rho_b^*
+
r_c\rho_c^*,
\]

\[
a_2
=
-
\left(
p_a r_a\rho_a^*\rho_b^*
+
p_b r_b\rho_b^*\rho_c^*
\right),
\]

and

\[
a_3
=
p_a p_b r_a
\rho_a^*\rho_b^*\rho_c^*.
\]

Since all parameters and equilibrium densities are positive,

\[
a_1>0,
\qquad
a_2<0,
\qquad
a_3>0.
\]

The Routh--Hurwitz criterion for a cubic polynomial requires

\[
a_1>0,
\qquad
a_2>0,
\qquad
a_3>0,
\qquad
a_1a_2>a_3.
\]

However, the coefficient $a_2$ is strictly negative. Therefore the Routh--Hurwitz conditions cannot be satisfied.

Consequently, at least one eigenvalue possesses a positive real part, implying that the interior coexistence equilibrium is always unstable whenever it exists.

Thus, {\it the interior coexistence equilibrium $E_{ABC}$ is unstable for all positive parameter values.} This result provides a simple analytical explanation for the emergence of lower-dimensional coexistence states. Rather than supporting stable coexistence of all three species, the hierarchical interaction structure destabilizes the interior equilibrium and drives the system toward coexistence among nonadjacent species or toward single-species dominance.

\section{Stability of Lower-Dimensional Equilibria for the Five-Species Hierarchical Model}
\label{Five-species}

In this appendix, we summarize the stability properties of the extinction, single-species, two-species, three-species, and four-species equilibrium states of the hierarchical May--Leonard model. The analytical results complement the numerical stability analyses presented in Figs.\ (3) and (4) of the main text.

\subsection{Extinction Equilibrium}

The extinction equilibrium is given by $(\rho_a^*,\rho_b^*,\rho_c^*,\rho_d^*,\rho_e^*)=(0,0,0,0,0).$ Linearization about this state yields a diagonal Jacobian with eigenvalues $\lambda_i = r_i-d_i,  i=a,b,c,d,e.$ Therefore, the extinction state is stable only when $r_i<d_i, \forall i$. Whenever at least one species satisfies $r_i>d_i$, the extinction equilibrium becomes unstable and invasion of the empty state becomes possible.

\begin{theorem}
Consider the five-species hierarchical model on the nonnegative orthant with all parameters \(r_i,p_i,d_i\) nonnegative and \(p_i>0\) for predation terms.
If
	\[
	r_i < d_i \qquad\text{for } i\in\{a,b,c,d,e\},
	\]
	then the origin \(\boldsymbol 0=(0,0,0,0,0)\) is globally asymptotically
	stable in \(\mathbb{R}_{\ge 0}^5\).
\end{theorem}

\begin{proof}
	Define the candidate Lyapunov function
	\[
	V(\rho_a,\rho_b,\rho_c,\rho_d,\rho_e):=\rho_a+\rho_b+\rho_c+\rho_d+\rho_e = S,
	\]
	which satisfies \(V\ge0\) on \(\mathbb{R}_{\ge 0}^5\) and \(V(\boldsymbol 0)=0\).
	
	Differentiate along trajectories:
	\[
	\begin{aligned}
		\dot V &= \dot{\rho}_a+\dot{\rho}_b+\dot{\rho}_c+\dot{\rho}_d+\dot{\rho}_e \\
		&= \rho_a\big(r_a(1-S)-d_a\big)
		+\rho_b\big(r_b(1-S)-p_a\rho_a-d_b\big)\\
		& +\rho_c\big(r_c(1-S)-p_b\rho_b-d_c\big)
		+\rho_d\big(r_d(1-S)-p_c\rho_c-d_d\big)\\
		&\quad +\rho_e\big(r_e(1-S)-p_d\rho_d-d_e\big).
	\end{aligned}
	\]
	Use \(1-S=\rho_v\le1\) and \(p_i\rho_j\ge0\) to bound each per-capita growth term by \(r_i-d_i\):
	\[
	\rho_i\big(r_i(1-S)-(\text{nonneg})-d_i\big) \le \rho_i(r_i-d_i),\qquad i\in\{a,\dots,e\}.
	\]
	Summing,
	\[
	\dot V \le \sum_{i\in\{a,\dots,e\}} \rho_i (r_i-d_i).
	\]
	Under the hypothesis \(r_i<d_i\) for every \(i\), each coefficient \((r_i-d_i)\) is strictly negative, so
	\[
	\dot V \le -\sum_{i} (d_i-r_i)\rho_i \le 0,
	\]
	with equality only when \(\rho_i=0\) for every \(i\). Thus \(\dot V<0\) for all \(\boldsymbol\rho\neq\boldsymbol0\) in \(\mathbb{R}_{\ge0}^5\).
	
	Hence \(V\) is a strict Lyapunov function on \(\mathbb{R}_{\ge0}^5\setminus\{\boldsymbol0\}\). By standard Lyapunov arguments, trajectories starting in \(\mathbb{R}_{\ge0}^5\) remain in the nonnegative orthant (forward invariance) and all trajectories in \(\mathbb{R}_{\ge0}^5\) converge to \(\boldsymbol 0\), i.e. the origin is globally asymptotically stable.
\end{proof}

Among the lower-dimensional equilibria, the extinction state admits a global stability result, whereas the remaining equilibria are characterized through local stability analysis.

\subsection{Single-Species Equilibria}

The model admits five single-species equilibria of the form $
\rho_i^*=1-\frac{d_i}{r_i},$ which are feasible whenever $r_i>d_i$. The equilibrium densities and corresponding eigenvalues of the Jacobian are summarized in Table~\ref{tab:single_species}.

\begin{table}[H]
\scriptsize
\centering
\caption{Single-species equilibria and corresponding eigenvalues.}
\label{tab:single_species}

\begin{tabular}{|c|p{5.5cm}|}
\hline
 Density & Eigenvalues \\
\hline

$\rho_a^*=1-\dfrac{d_a}{r_a}$
&
$\lambda_1=-r_a\rho_a^*$,
$\lambda_2=r_b(1-\rho_a^*)-p_a\rho_a^*-d_b$,
$\lambda_3=r_c(1-\rho_a^*)-d_c$,
$\lambda_4=r_d(1-\rho_a^*)-d_d$,
$\lambda_5=r_e(1-\rho_a^*)-d_e$
\\
\hline

$\rho_b^*=1-\dfrac{d_b}{r_b}$
&
$\lambda_1=r_a-d_a$,
$\lambda_2=-r_b\rho_b^*$,
$\lambda_3=r_c(1-\rho_b^*)-p_b\rho_b^*-d_c$,
$\lambda_4=r_d(1-\rho_b^*)-d_d$,
$\lambda_5=r_e(1-\rho_b^*)-d_e$
\\
\hline

$\rho_c^*=1-\dfrac{d_c}{r_c}$
&
$\lambda_1=r_a-d_a$,
$\lambda_2=r_b-d_b$,
$\lambda_3=-r_c\rho_c^*$,
$\lambda_4=r_d(1-\rho_c^*)-p_c\rho_c^*-d_d$,
$\lambda_5=r_e(1-\rho_c^*)-d_e$
\\
\hline

$\rho_d^*=1-\dfrac{d_d}{r_d}$
&
$\lambda_1=r_a-d_a$,
$\lambda_2=r_b-d_b$,
$\lambda_3=r_c-d_c$,
$\lambda_4=-r_d\rho_d^*$,
$\lambda_5=r_e(1-\rho_d^*)-p_d\rho_d^*-d_e$
\\
\hline

$\rho_e^*=1-\dfrac{d_e}{r_e}$
&
$\lambda_1=r_a-d_a$,
$\lambda_2=r_b-d_b$,
$\lambda_3=r_c-d_c$,
$\lambda_4=r_d-d_d$,
$\lambda_5=-r_e\rho_e^*$
\\
\hline

\end{tabular}
\end{table}









\vspace{2mm}




























The stability of each equilibrium is determined by the sign of the largest eigenvalue. For the parameter set considered in the main text, the dominant eigenvalue was evaluated numerically as a function of $r_a$, and the resulting stability transitions are shown in Fig.~\ref{fig:2}.

\subsection{Two-Species Coexistence Equilibria ($r_i=r,\; d_i=d$)}

The non-adjacent two-species coexistence states form one-dimensional manifolds satisfying

\[
\rho_i^*+\rho_j^*=1-\frac{d}{r}.
\]

The corresponding eigenvalue spectra are summarized in Table~\ref{tab:two_species}.

\begin{table}[H]
\small
\centering
\caption{Two-species coexistence equilibria and corresponding eigenvalue spectra for the symmetric case $r_i=r$ and $d_i=d$.}
\label{tab:two_species}

\begin{tabular}{|c|c|p{4.0cm}|}
\hline
State & Equilibrium condition & Eigenvalues \\
\hline

$(A,C)$
&
$\rho_a^*+\rho_c^*=1-\dfrac{d}{r}$
&
$\{0,\;d-r,\;-p_a\rho_a^*,\;-p_c\rho_c^*,\;0\}$
\\
\hline

$(A,D)$
&
$\rho_a^*+\rho_d^*=1-\dfrac{d}{r}$
&
$\{0,\;d-r,\;-p_a\rho_a^*,\;-p_d\rho_d^*,\;0\}$
\\
\hline

$(A,E)$
&
$\rho_a^*+\rho_e^*=1-\dfrac{d}{r}$
&
$\{0,\;d-r,\;-p_a\rho_a^*,\;0,\;0\}$
\\
\hline

$(B,D)$
&
$\rho_b^*+\rho_d^*=1-\dfrac{d}{r}$
&
$\{0,\;d-r,\;-p_b\rho_b^*,\;-p_d\rho_d^*,\;0\}$
\\
\hline

$(B,E)$
&
$\rho_b^*+\rho_e^*=1-\dfrac{d}{r}$
&
$\{0,\;d-r,\;-p_b\rho_b^*,\;0,\;0\}$
\\
\hline

$(C,E)$
&
$\rho_c^*+\rho_e^*=1-\dfrac{d}{r}$
&
$\{0,\;d-r,\;-p_c\rho_c^*,\;0,\;0\}$
\\
\hline

\end{tabular}
\end{table}

For all coexistence states, one zero eigenvalue arises from the continuous line of equilibria. The eigenvalue $d-r$ determines stability transverse to the coexistence manifold and is negative whenever $r>d$. The remaining nonzero eigenvalues are associated with predation-induced damping and are strictly negative for positive densities. Consequently, these coexistence manifolds are neutrally stable along the coexistence direction and locally attracting in the transverse directions when $r>d$.

\subsection{Three-Species Coexistence Equilibria ($r_i=r, d_i=d$)}

The only feasible three-species coexistence state in the symmetric model is (A,C,E), for which $\rho_a^*+\rho_c^*+\rho_e^*=1-\frac{d}{r},
\qquad
\rho_b^*=\rho_d^*=0.$

The Jacobian evaluated at this equilibrium possesses the eigenvalue spectrum

$$
\mathrm{Spec}(J^*)
=\left\{
0,
-p_a\rho_a^*,
d-r,
-p_c\rho_c^*,
0
\right\}.$$

Since $p_a>0, p_c>0,$ the eigenvalues $-p_a\rho_a^*, -p_c\rho_c^*$ are strictly negative for positive equilibrium densities. Furthermore, $d-r<0$ whenever $r>d$. The two zero eigenvalues arise from the continuous family of equilibria generated by the constraint

$$\rho_a^*+\rho_c^*+\rho_e^*=1-\frac{d}{r},$$

and correspond to neutral motion along the coexistence manifold. Consequently, the (A,C,E) coexistence state is neutrally stable along the manifold and locally stable in all transverse directions when $r>d$.

\subsection{Nonexistence of Four-Species Coexistence Equilibria ($r_i=r,; d_i=d$)}

In the symmetric case $r_i=r$ and $d_i=d$, four-species coexistence equilibria cannot exist. To demonstrate this, suppose that four species coexist at equilibrium. Without loss of generality, consider the state $\rho_a^*,\rho_b^*,\rho_c^*,\rho_d^*>0, \rho_e^*=0.$

Since $\rho_a^*>0$, the steady-state condition $\dot{\rho}_a=0$ yields $r(1-S)-d=0,$ and therefore $S=1-\frac{d}{r}.$

Substituting this relation into the equilibrium condition for species B,
we have $$0=r(1-S)-p_a\rho_a^*-d,$$ gives

$$0=(r(1-S)-d)-p_a\rho_a^*
=-p_a\rho_a^*.$$

Since $p_a>0$, it follows that $\rho_a^*=0,$ which contradicts the assumption $\rho_a^*>0$.

The same argument applies to any four-species configuration. The hierarchical predation terms impose a chain of exclusion in which the coexistence conditions necessarily force at least one density to vanish. Consequently, no equilibrium with four simultaneously surviving species can satisfy the steady-state equations.

Thus, {\it no four-species coexistence equilibrium exists when $r_i=r, d_i=d.$} This result highlights a fundamental consequence of the hierarchical interaction structure: although two- and three-species coexistence states are permitted, the symmetric system cannot support coexistence of four species at equilibrium.

\subsection{Summary}

The stability analysis reveals a clear hierarchy among equilibrium states. The interior coexistence equilibrium is generically unstable, while lower-dimensional coexistence states may become stable depending on parameter values. Stable coexistence is preferentially associated with non-adjacent species, particularly the (A,C,E) configuration, whereas coexistence among directly interacting species is generally suppressed by the hierarchical predation structure. These results provide the analytical basis for the coexistence-selection rules discussed in the main text.

\section{Monte Carlo simulation in five-species system}
\label{monte-carlo}

The Monte Carlo simulation is performed on a system implemented on a two-dimensional square lattice of size $L \times L$ with periodic boundary conditions, where each lattice site can either be vacant or occupied by a single individual belonging to one of the five species. The initial configuration is generated randomly by assigning each site to one of the five species or to a vacant state according to the densities which ensure that the normalization condition,
$\rho_a+\rho_b+\rho_c+\rho_d+\rho_e+\rho_v=1$ is satisfied. Once the condition is imposed in the initial configuration, it is automatically conserved in all the subsequent Monte Carlo steps.

The dynamics evolve through asynchronous random sequential updates. During each elementary update, a lattice site is selected uniformly at random. If the selected site is occupied, one of its four nearest neighbours in the von Neumann neighbourhood is chosen randomly, and the interaction between the two sites is executed according to the stochastic reaction rules of the model. Depending on the interacting species, the selected individual may eliminate its prey through predation, reproduce into a neighbouring vacant site,  or die spontaneously, with probabilities proportional to the corresponding reaction rates. Each successful event modifies the local lattice configuration while preserving the single-occupancy constraint. One time unit is defined by $L^2$ Monte Carlo steps of elementary update attempts so that, on average, every lattice site is selected once per unit time.

 Ensemble averages are obtained by performing simulations on a large number of independent realizations with different random initial configurations and random number sequences. Species densities are measured and averaged over the steady-state regime. The steady state is obtained after discarding an initial transient eliminating finite-time relaxation effects. The stochastic simulations thus capture the effects of spatial correlations, fluctuations, and finite-size noise that are absent in the mean-field description.

}

\section{Author Contribution}
\noindent
{\bf Rakesh Samanta}: Data curation; Formal analysis; Investigation; Validation; Writing-original draft; Writing-review \& editing. {\bf Shraosi Dawn}: Data curation; Formal analysis; Investigation; Validation; Writing-original draft; Writing-review \& editing. {\bf Sk Jahiruddin}: Data curation; Formal analysis; Investigation; Validation. {\bf Sirshendu Bhattacharyya}: Conceptualization; Formal analysis; Investigation; Methodology; Project administration; Supervision; Validation; Visualization; Writing-original draft; Writing-review \& editing. {\bf Chittaranjan Hens}: Conceptualization; Formal analysis; Investigation; Methodology; Project administration; Supervision; Validation; Visualization; Writing-original draft; Writing-review \& editing. {\bf Sayantan Nag Chowdhury}: Conceptualization; Formal analysis; Investigation; Methodology; Project administration; Supervision; Validation; Visualization; Writing-original draft; Writing-review \& editing.

\bibliography{rps_death}

\begin{thebibliography}{55}%
\makeatletter
\providecommand \@ifxundefined [1]{%
 \@ifx{#1\undefined}
}%
\providecommand \@ifnum [1]{%
 \ifnum #1\expandafter \@firstoftwo
 \else \expandafter \@secondoftwo
 \fi
}%
\providecommand \@ifx [1]{%
 \ifx #1\expandafter \@firstoftwo
 \else \expandafter \@secondoftwo
 \fi
}%
\providecommand \natexlab [1]{#1}%
\providecommand \enquote  [1]{``#1''}%
\providecommand \bibnamefont  [1]{#1}%
\providecommand \bibfnamefont [1]{#1}%
\providecommand \citenamefont [1]{#1}%
\providecommand \href@noop [0]{\@secondoftwo}%
\providecommand \href [0]{\begingroup \@sanitize@url \@href}%
\providecommand \@href[1]{\@@startlink{#1}\@@href}%
\providecommand \@@href[1]{\endgroup#1\@@endlink}%
\providecommand \@sanitize@url [0]{\catcode `\\12\catcode `\$12\catcode
  `\&12\catcode `\#12\catcode `\^12\catcode `\_12\catcode `\%12\relax}%
\providecommand \@@startlink[1]{}%
\providecommand \@@endlink[0]{}%
\providecommand \url  [0]{\begingroup\@sanitize@url \@url }%
\providecommand \@url [1]{\endgroup\@href {#1}{\urlprefix }}%
\providecommand \urlprefix  [0]{URL }%
\providecommand \Eprint [0]{\href }%
\providecommand \doibase [0]{http://dx.doi.org/}%
\providecommand \selectlanguage [0]{\@gobble}%
\providecommand \bibinfo  [0]{\@secondoftwo}%
\providecommand \bibfield  [0]{\@secondoftwo}%
\providecommand \translation [1]{[#1]}%
\providecommand \BibitemOpen [0]{}%
\providecommand \bibitemStop [0]{}%
\providecommand \bibitemNoStop [0]{.\EOS\space}%
\providecommand \EOS [0]{\spacefactor3000\relax}%
\providecommand \BibitemShut  [1]{\csname bibitem#1\endcsname}%
\let\auto@bib@innerbib\@empty
\bibitem [{\citenamefont {Lotka}(1920)}]{lotka1920analytical}%
  \BibitemOpen
  \bibfield  {author} {\bibinfo {author} {\bibfnamefont {A.~J.}\ \bibnamefont
  {Lotka}},\ }\href@noop {} {\bibfield  {journal} {\bibinfo  {journal}
  {Proceedings of the National Academy of Sciences}\ }\textbf {\bibinfo
  {volume} {6}},\ \bibinfo {pages} {410} (\bibinfo {year} {1920})}\BibitemShut
  {NoStop}%
\bibitem [{\citenamefont {Volterra}(1927)}]{volterra1927fluctuations}%
  \BibitemOpen
  \bibfield  {author} {\bibinfo {author} {\bibfnamefont {V.}~\bibnamefont
  {Volterra}},\ }\href@noop {} {\bibfield  {journal} {\bibinfo  {journal}
  {Nature}\ }\textbf {\bibinfo {volume} {119}},\ \bibinfo {pages} {12}
  (\bibinfo {year} {1927})}\BibitemShut {NoStop}%
\bibitem [{\citenamefont {May}\ and\ \citenamefont
  {Leonard}(1975)}]{may1975nonlinear}%
  \BibitemOpen
  \bibfield  {author} {\bibinfo {author} {\bibfnamefont {R.~M.}\ \bibnamefont
  {May}}\ and\ \bibinfo {author} {\bibfnamefont {W.~J.}\ \bibnamefont
  {Leonard}},\ }\href@noop {} {\bibfield  {journal} {\bibinfo  {journal} {SIAM
  Journal on Applied Mathematics}\ }\textbf {\bibinfo {volume} {29}},\ \bibinfo
  {pages} {243} (\bibinfo {year} {1975})}\BibitemShut {NoStop}%
\bibitem [{\citenamefont {He}\ \emph {et~al.}(2010)\citenamefont {He},
  \citenamefont {Mobilia},\ and\ \citenamefont {T{\"a}uber}}]{he2010spatial}%
  \BibitemOpen
  \bibfield  {author} {\bibinfo {author} {\bibfnamefont {Q.}~\bibnamefont
  {He}}, \bibinfo {author} {\bibfnamefont {M.}~\bibnamefont {Mobilia}}, \ and\
  \bibinfo {author} {\bibfnamefont {U.~C.}\ \bibnamefont {T{\"a}uber}},\
  }\href@noop {} {\bibfield  {journal} {\bibinfo  {journal} {Physical Review
  E}\ }\textbf {\bibinfo {volume} {82}},\ \bibinfo {pages} {051909} (\bibinfo
  {year} {2010})}\BibitemShut {NoStop}%
\bibitem [{\citenamefont {Jiang}\ \emph {et~al.}(2012)\citenamefont {Jiang},
  \citenamefont {Wang}, \citenamefont {Lai},\ and\ \citenamefont
  {Ni}}]{jiang2012multi}%
  \BibitemOpen
  \bibfield  {author} {\bibinfo {author} {\bibfnamefont {L.-L.}\ \bibnamefont
  {Jiang}}, \bibinfo {author} {\bibfnamefont {W.-X.}\ \bibnamefont {Wang}},
  \bibinfo {author} {\bibfnamefont {Y.-C.}\ \bibnamefont {Lai}}, \ and\
  \bibinfo {author} {\bibfnamefont {X.}~\bibnamefont {Ni}},\ }\href@noop {}
  {\bibfield  {journal} {\bibinfo  {journal} {Physics Letters A}\ }\textbf
  {\bibinfo {volume} {376}},\ \bibinfo {pages} {2292} (\bibinfo {year}
  {2012})}\BibitemShut {NoStop}%
\bibitem [{\citenamefont {Huang}\ \emph {et~al.}(2023)\citenamefont {Huang},
  \citenamefont {Duan}, \citenamefont {Qin},\ and\ \citenamefont
  {Park}}]{huang2023fitness}%
  \BibitemOpen
  \bibfield  {author} {\bibinfo {author} {\bibfnamefont {W.}~\bibnamefont
  {Huang}}, \bibinfo {author} {\bibfnamefont {X.}~\bibnamefont {Duan}},
  \bibinfo {author} {\bibfnamefont {L.}~\bibnamefont {Qin}}, \ and\ \bibinfo
  {author} {\bibfnamefont {J.}~\bibnamefont {Park}},\ }\href@noop {} {\bibfield
   {journal} {\bibinfo  {journal} {Applied Mathematics and Computation}\
  }\textbf {\bibinfo {volume} {456}},\ \bibinfo {pages} {128135} (\bibinfo
  {year} {2023})}\BibitemShut {NoStop}%
\bibitem [{\citenamefont {Menezes}\ and\ \citenamefont
  {Tenorio}(2023)}]{menezes2023spatial}%
  \BibitemOpen
  \bibfield  {author} {\bibinfo {author} {\bibfnamefont {J.}~\bibnamefont
  {Menezes}}\ and\ \bibinfo {author} {\bibfnamefont {M.}~\bibnamefont
  {Tenorio}},\ }\href@noop {} {\bibfield  {journal} {\bibinfo  {journal}
  {Journal of Physics: Complexity}\ }\textbf {\bibinfo {volume} {4}},\ \bibinfo
  {pages} {025015} (\bibinfo {year} {2023})}\BibitemShut {NoStop}%
\bibitem [{\citenamefont {Cheng}\ \emph {et~al.}(2014)\citenamefont {Cheng},
  \citenamefont {Yao}, \citenamefont {Huang}, \citenamefont {Park},
  \citenamefont {Do},\ and\ \citenamefont {Lai}}]{cheng2014mesoscopic}%
  \BibitemOpen
  \bibfield  {author} {\bibinfo {author} {\bibfnamefont {H.}~\bibnamefont
  {Cheng}}, \bibinfo {author} {\bibfnamefont {N.}~\bibnamefont {Yao}}, \bibinfo
  {author} {\bibfnamefont {Z.-G.}\ \bibnamefont {Huang}}, \bibinfo {author}
  {\bibfnamefont {J.}~\bibnamefont {Park}}, \bibinfo {author} {\bibfnamefont
  {Y.}~\bibnamefont {Do}}, \ and\ \bibinfo {author} {\bibfnamefont {Y.-C.}\
  \bibnamefont {Lai}},\ }\href@noop {} {\bibfield  {journal} {\bibinfo
  {journal} {Scientific Reports}\ }\textbf {\bibinfo {volume} {4}},\ \bibinfo
  {pages} {1} (\bibinfo {year} {2014})}\BibitemShut {NoStop}%
\bibitem [{\citenamefont {Park}\ \emph {et~al.}(2017)\citenamefont {Park},
  \citenamefont {Do}, \citenamefont {Jang},\ and\ \citenamefont
  {Lai}}]{park2017emergence}%
  \BibitemOpen
  \bibfield  {author} {\bibinfo {author} {\bibfnamefont {J.}~\bibnamefont
  {Park}}, \bibinfo {author} {\bibfnamefont {Y.}~\bibnamefont {Do}}, \bibinfo
  {author} {\bibfnamefont {B.}~\bibnamefont {Jang}}, \ and\ \bibinfo {author}
  {\bibfnamefont {Y.-C.}\ \bibnamefont {Lai}},\ }\href@noop {} {\bibfield
  {journal} {\bibinfo  {journal} {Scientific Reports}\ }\textbf {\bibinfo
  {volume} {7}},\ \bibinfo {pages} {1} (\bibinfo {year} {2017})}\BibitemShut
  {NoStop}%
\bibitem [{\citenamefont {Wang}\ \emph {et~al.}(2022)\citenamefont {Wang},
  \citenamefont {Lu}, \citenamefont {Shi},\ and\ \citenamefont
  {Park}}]{wang2022effect}%
  \BibitemOpen
  \bibfield  {author} {\bibinfo {author} {\bibfnamefont {X.}~\bibnamefont
  {Wang}}, \bibinfo {author} {\bibfnamefont {Y.}~\bibnamefont {Lu}}, \bibinfo
  {author} {\bibfnamefont {L.}~\bibnamefont {Shi}}, \ and\ \bibinfo {author}
  {\bibfnamefont {J.}~\bibnamefont {Park}},\ }\href@noop {} {\bibfield
  {journal} {\bibinfo  {journal} {Scientific Reports}\ }\textbf {\bibinfo
  {volume} {12}},\ \bibinfo {pages} {1821} (\bibinfo {year}
  {2022})}\BibitemShut {NoStop}%
\bibitem [{\citenamefont {Avelino}\ \emph {et~al.}(2012)\citenamefont
  {Avelino}, \citenamefont {Bazeia}, \citenamefont {Losano},\ and\
  \citenamefont {Menezes}}]{avelino2012neummann}%
  \BibitemOpen
  \bibfield  {author} {\bibinfo {author} {\bibfnamefont {P.}~\bibnamefont
  {Avelino}}, \bibinfo {author} {\bibfnamefont {D.}~\bibnamefont {Bazeia}},
  \bibinfo {author} {\bibfnamefont {L.}~\bibnamefont {Losano}}, \ and\ \bibinfo
  {author} {\bibfnamefont {J.}~\bibnamefont {Menezes}},\ }\href@noop {}
  {\bibfield  {journal} {\bibinfo  {journal} {Physical Review E—Statistical,
  Nonlinear, and Soft Matter Physics}\ }\textbf {\bibinfo {volume} {86}},\
  \bibinfo {pages} {031119} (\bibinfo {year} {2012})}\BibitemShut {NoStop}%
\bibitem [{\citenamefont {Mohd}\ and\ \citenamefont
  {Park}(2021)}]{mohd2021interplay}%
  \BibitemOpen
  \bibfield  {author} {\bibinfo {author} {\bibfnamefont {M.~H.}\ \bibnamefont
  {Mohd}}\ and\ \bibinfo {author} {\bibfnamefont {J.}~\bibnamefont {Park}},\
  }\href@noop {} {\bibfield  {journal} {\bibinfo  {journal} {Chaos, Solitons \&
  Fractals}\ }\textbf {\bibinfo {volume} {153}},\ \bibinfo {pages} {111579}
  (\bibinfo {year} {2021})}\BibitemShut {NoStop}%
\bibitem [{\citenamefont {Wangersky}(1978)}]{wangersky1978lotka}%
  \BibitemOpen
  \bibfield  {author} {\bibinfo {author} {\bibfnamefont {P.~J.}\ \bibnamefont
  {Wangersky}},\ }\href@noop {} {\bibfield  {journal} {\bibinfo  {journal}
  {Annual Review of Ecology and Systematics}\ }\textbf {\bibinfo {volume}
  {9}},\ \bibinfo {pages} {189} (\bibinfo {year} {1978})}\BibitemShut {NoStop}%
\bibitem [{\citenamefont {Anisiu}(2014)}]{anisiu2014lotka}%
  \BibitemOpen
  \bibfield  {author} {\bibinfo {author} {\bibfnamefont {M.-C.}\ \bibnamefont
  {Anisiu}},\ }\href@noop {} {\bibfield  {journal} {\bibinfo  {journal}
  {Did{\'a}ctica Mathematica}\ }\textbf {\bibinfo {volume} {32}} (\bibinfo
  {year} {2014})}\BibitemShut {NoStop}%
\bibitem [{\citenamefont {Bunin}(2017)}]{bunin2017ecological}%
  \BibitemOpen
  \bibfield  {author} {\bibinfo {author} {\bibfnamefont {G.}~\bibnamefont
  {Bunin}},\ }\href@noop {} {\bibfield  {journal} {\bibinfo  {journal}
  {Physical Review E}\ }\textbf {\bibinfo {volume} {95}},\ \bibinfo {pages}
  {042414} (\bibinfo {year} {2017})}\BibitemShut {NoStop}%
\bibitem [{\citenamefont {Din}(2013)}]{din2013dynamics}%
  \BibitemOpen
  \bibfield  {author} {\bibinfo {author} {\bibfnamefont {Q.}~\bibnamefont
  {Din}},\ }\href@noop {} {\bibfield  {journal} {\bibinfo  {journal} {Advances
  in Difference Equations}\ }\textbf {\bibinfo {volume} {2013}},\ \bibinfo
  {pages} {95} (\bibinfo {year} {2013})}\BibitemShut {NoStop}%
\bibitem [{\citenamefont {Takeuchi}(1996)}]{takeuchi1996global}%
  \BibitemOpen
  \bibfield  {author} {\bibinfo {author} {\bibfnamefont {Y.}~\bibnamefont
  {Takeuchi}},\ }\href@noop {} {\emph {\bibinfo {title} {Global dynamical
  properties of Lotka-Volterra systems}}}\ (\bibinfo  {publisher} {World
  Scientific},\ \bibinfo {year} {1996})\BibitemShut {NoStop}%
\bibitem [{\citenamefont {Hern{\'a}ndez-Bermejo}\ and\ \citenamefont
  {Fair{\'e}n}(1997)}]{hernandez1997lotka}%
  \BibitemOpen
  \bibfield  {author} {\bibinfo {author} {\bibfnamefont {B.}~\bibnamefont
  {Hern{\'a}ndez-Bermejo}}\ and\ \bibinfo {author} {\bibfnamefont
  {V.}~\bibnamefont {Fair{\'e}n}},\ }\href@noop {} {\bibfield  {journal}
  {\bibinfo  {journal} {Mathematical Biosciences}\ }\textbf {\bibinfo {volume}
  {140}},\ \bibinfo {pages} {1} (\bibinfo {year} {1997})}\BibitemShut {NoStop}%
\bibitem [{\citenamefont {Serrao}\ and\ \citenamefont
  {T{\"a}uber}(2017)}]{serrao2017stochastic}%
  \BibitemOpen
  \bibfield  {author} {\bibinfo {author} {\bibfnamefont {S.~R.}\ \bibnamefont
  {Serrao}}\ and\ \bibinfo {author} {\bibfnamefont {U.~C.}\ \bibnamefont
  {T{\"a}uber}},\ }\href@noop {} {\bibfield  {journal} {\bibinfo  {journal}
  {Journal of Physics A: Mathematical and Theoretical}\ }\textbf {\bibinfo
  {volume} {50}},\ \bibinfo {pages} {404005} (\bibinfo {year}
  {2017})}\BibitemShut {NoStop}%
\bibitem [{\citenamefont {Chi}\ \emph {et~al.}(1998)\citenamefont {Chi},
  \citenamefont {Wu},\ and\ \citenamefont {Hsu}}]{chi1998asymmetric}%
  \BibitemOpen
  \bibfield  {author} {\bibinfo {author} {\bibfnamefont {C.-W.}\ \bibnamefont
  {Chi}}, \bibinfo {author} {\bibfnamefont {L.-I.}\ \bibnamefont {Wu}}, \ and\
  \bibinfo {author} {\bibfnamefont {S.-B.}\ \bibnamefont {Hsu}},\ }\href@noop
  {} {\bibfield  {journal} {\bibinfo  {journal} {SIAM Journal on Applied
  Mathematics}\ }\textbf {\bibinfo {volume} {58}},\ \bibinfo {pages} {211}
  (\bibinfo {year} {1998})}\BibitemShut {NoStop}%
\bibitem [{\citenamefont {Barendregt}\ and\ \citenamefont
  {Thomas}(2023)}]{barendregt2023heteroclinic}%
  \BibitemOpen
  \bibfield  {author} {\bibinfo {author} {\bibfnamefont {N.~W.}\ \bibnamefont
  {Barendregt}}\ and\ \bibinfo {author} {\bibfnamefont {P.~J.}\ \bibnamefont
  {Thomas}},\ }\href@noop {} {\bibfield  {journal} {\bibinfo  {journal}
  {Journal of Mathematical Biology}\ }\textbf {\bibinfo {volume} {86}},\
  \bibinfo {pages} {30} (\bibinfo {year} {2023})}\BibitemShut {NoStop}%
\bibitem [{\citenamefont {Leonard}\ and\ \citenamefont
  {May}(1975)}]{leonard1975nonlinear}%
  \BibitemOpen
  \bibfield  {author} {\bibinfo {author} {\bibfnamefont {W.}~\bibnamefont
  {Leonard}}\ and\ \bibinfo {author} {\bibfnamefont {R.}~\bibnamefont {May}},\
  }\href@noop {} {\bibfield  {journal} {\bibinfo  {journal} {SIAM J. Appl.
  Math}\ }\textbf {\bibinfo {volume} {29}},\ \bibinfo {pages} {243} (\bibinfo
  {year} {1975})}\BibitemShut {NoStop}%
\bibitem [{\citenamefont {Sinervo}\ and\ \citenamefont
  {Lively}(1996)}]{sinervo1996rock}%
  \BibitemOpen
  \bibfield  {author} {\bibinfo {author} {\bibfnamefont {B.}~\bibnamefont
  {Sinervo}}\ and\ \bibinfo {author} {\bibfnamefont {C.~M.}\ \bibnamefont
  {Lively}},\ }\href@noop {} {\bibfield  {journal} {\bibinfo  {journal}
  {Nature}\ }\textbf {\bibinfo {volume} {380}},\ \bibinfo {pages} {240}
  (\bibinfo {year} {1996})}\BibitemShut {NoStop}%
\bibitem [{\citenamefont {Szab{\'o}}\ and\ \citenamefont
  {Fath}(2007)}]{szabo2007evolutionary}%
  \BibitemOpen
  \bibfield  {author} {\bibinfo {author} {\bibfnamefont {G.}~\bibnamefont
  {Szab{\'o}}}\ and\ \bibinfo {author} {\bibfnamefont {G.}~\bibnamefont
  {Fath}},\ }\href@noop {} {\bibfield  {journal} {\bibinfo  {journal} {Physics
  Reports}\ }\textbf {\bibinfo {volume} {446}},\ \bibinfo {pages} {97}
  (\bibinfo {year} {2007})}\BibitemShut {NoStop}%
\bibitem [{\citenamefont {Wu}\ \emph {et~al.}(2010)\citenamefont {Wu},
  \citenamefont {Altrock}, \citenamefont {Wang},\ and\ \citenamefont
  {Traulsen}}]{wu2010}%
  \BibitemOpen
  \bibfield  {author} {\bibinfo {author} {\bibfnamefont {B.}~\bibnamefont
  {Wu}}, \bibinfo {author} {\bibfnamefont {P.~M.}\ \bibnamefont {Altrock}},
  \bibinfo {author} {\bibfnamefont {L.}~\bibnamefont {Wang}}, \ and\ \bibinfo
  {author} {\bibfnamefont {A.}~\bibnamefont {Traulsen}},\ }\href {\doibase
  10.1103/PhysRevE.82.046106} {\bibfield  {journal} {\bibinfo  {journal} {Phys.
  Rev. E}\ }\textbf {\bibinfo {volume} {82}},\ \bibinfo {pages} {046106}
  (\bibinfo {year} {2010})}\BibitemShut {NoStop}%
\bibitem [{\citenamefont {Claussen}\ and\ \citenamefont
  {Traulsen}(2008)}]{claussen2008cyclic}%
  \BibitemOpen
  \bibfield  {author} {\bibinfo {author} {\bibfnamefont {J.~C.}\ \bibnamefont
  {Claussen}}\ and\ \bibinfo {author} {\bibfnamefont {A.}~\bibnamefont
  {Traulsen}},\ }\href@noop {} {\bibfield  {journal} {\bibinfo  {journal}
  {Physical Review Letters}\ }\textbf {\bibinfo {volume} {100}},\ \bibinfo
  {pages} {058104} (\bibinfo {year} {2008})}\BibitemShut {NoStop}%
\bibitem [{\citenamefont {Kerr}\ \emph {et~al.}(2002)\citenamefont {Kerr},
  \citenamefont {Riley}, \citenamefont {Feldman},\ and\ \citenamefont
  {Bohannan}}]{kerr2002local}%
  \BibitemOpen
  \bibfield  {author} {\bibinfo {author} {\bibfnamefont {B.}~\bibnamefont
  {Kerr}}, \bibinfo {author} {\bibfnamefont {M.~A.}\ \bibnamefont {Riley}},
  \bibinfo {author} {\bibfnamefont {M.~W.}\ \bibnamefont {Feldman}}, \ and\
  \bibinfo {author} {\bibfnamefont {B.~J.}\ \bibnamefont {Bohannan}},\
  }\href@noop {} {\bibfield  {journal} {\bibinfo  {journal} {Nature}\ }\textbf
  {\bibinfo {volume} {418}},\ \bibinfo {pages} {171} (\bibinfo {year}
  {2002})}\BibitemShut {NoStop}%
\bibitem [{\citenamefont {Reichenbach}\ \emph {et~al.}(2007)\citenamefont
  {Reichenbach}, \citenamefont {Mobilia},\ and\ \citenamefont
  {Frey}}]{reichenbach2007mobility}%
  \BibitemOpen
  \bibfield  {author} {\bibinfo {author} {\bibfnamefont {T.}~\bibnamefont
  {Reichenbach}}, \bibinfo {author} {\bibfnamefont {M.}~\bibnamefont
  {Mobilia}}, \ and\ \bibinfo {author} {\bibfnamefont {E.}~\bibnamefont
  {Frey}},\ }\href@noop {} {\bibfield  {journal} {\bibinfo  {journal} {Nature}\
  }\textbf {\bibinfo {volume} {448}},\ \bibinfo {pages} {1046} (\bibinfo {year}
  {2007})}\BibitemShut {NoStop}%
\bibitem [{\citenamefont {Reichenbach}\ and\ \citenamefont
  {Frey}(2008)}]{reichenbach2008instability}%
  \BibitemOpen
  \bibfield  {author} {\bibinfo {author} {\bibfnamefont {T.}~\bibnamefont
  {Reichenbach}}\ and\ \bibinfo {author} {\bibfnamefont {E.}~\bibnamefont
  {Frey}},\ }\href@noop {} {\bibfield  {journal} {\bibinfo  {journal} {Physical
  Review Letters}\ }\textbf {\bibinfo {volume} {101}},\ \bibinfo {pages}
  {058102} (\bibinfo {year} {2008})}\BibitemShut {NoStop}%
\bibitem [{\citenamefont {Mobilia}(2010)}]{mobiia2010jtb}%
  \BibitemOpen
  \bibfield  {author} {\bibinfo {author} {\bibfnamefont {M.}~\bibnamefont
  {Mobilia}},\ }\href {\doibase https://doi.org/10.1016/j.jtbi.2010.01.008}
  {\bibfield  {journal} {\bibinfo  {journal} {Journal of Theoretical Biology}\
  }\textbf {\bibinfo {volume} {264}},\ \bibinfo {pages} {1} (\bibinfo {year}
  {2010})}\BibitemShut {NoStop}%
\bibitem [{\citenamefont {Bhattacharyya}\ \emph {et~al.}(2020)\citenamefont
  {Bhattacharyya}, \citenamefont {Sinha}, \citenamefont {De},\ and\
  \citenamefont {Hens}}]{bhattacharyya2020pre}%
  \BibitemOpen
  \bibfield  {author} {\bibinfo {author} {\bibfnamefont {S.}~\bibnamefont
  {Bhattacharyya}}, \bibinfo {author} {\bibfnamefont {P.}~\bibnamefont
  {Sinha}}, \bibinfo {author} {\bibfnamefont {R.}~\bibnamefont {De}}, \ and\
  \bibinfo {author} {\bibfnamefont {C.}~\bibnamefont {Hens}},\ }\href {\doibase
  10.1103/PhysRevE.102.012220} {\bibfield  {journal} {\bibinfo  {journal}
  {Phys. Rev. E}\ }\textbf {\bibinfo {volume} {102}},\ \bibinfo {pages}
  {012220} (\bibinfo {year} {2020})}\BibitemShut {NoStop}%
\bibitem [{\citenamefont {Islam}\ \emph {et~al.}(2022)\citenamefont {Islam},
  \citenamefont {Mondal}, \citenamefont {Mobilia}, \citenamefont
  {Bhattacharyya},\ and\ \citenamefont {Hens}}]{islam2022pre}%
  \BibitemOpen
  \bibfield  {author} {\bibinfo {author} {\bibfnamefont {S.}~\bibnamefont
  {Islam}}, \bibinfo {author} {\bibfnamefont {A.}~\bibnamefont {Mondal}},
  \bibinfo {author} {\bibfnamefont {M.}~\bibnamefont {Mobilia}}, \bibinfo
  {author} {\bibfnamefont {S.}~\bibnamefont {Bhattacharyya}}, \ and\ \bibinfo
  {author} {\bibfnamefont {C.}~\bibnamefont {Hens}},\ }\href {\doibase
  10.1103/PhysRevE.105.014215} {\bibfield  {journal} {\bibinfo  {journal}
  {Phys. Rev. E}\ }\textbf {\bibinfo {volume} {105}},\ \bibinfo {pages}
  {014215} (\bibinfo {year} {2022})}\BibitemShut {NoStop}%
\bibitem [{\citenamefont {Juul}\ \emph {et~al.}(2012)\citenamefont {Juul},
  \citenamefont {Sneppen},\ and\ \citenamefont {Mathiesen}}]{juul2012pre}%
  \BibitemOpen
  \bibfield  {author} {\bibinfo {author} {\bibfnamefont {J.}~\bibnamefont
  {Juul}}, \bibinfo {author} {\bibfnamefont {K.}~\bibnamefont {Sneppen}}, \
  and\ \bibinfo {author} {\bibfnamefont {J.}~\bibnamefont {Mathiesen}},\ }\href
  {\doibase 10.1103/PhysRevE.85.061924} {\bibfield  {journal} {\bibinfo
  {journal} {Phys. Rev. E}\ }\textbf {\bibinfo {volume} {85}},\ \bibinfo
  {pages} {061924} (\bibinfo {year} {2012})}\BibitemShut {NoStop}%
\bibitem [{\citenamefont {Szolnoki}\ and\ \citenamefont
  {Perc}(2009)}]{szolnoki2009njp}%
  \BibitemOpen
  \bibfield  {author} {\bibinfo {author} {\bibfnamefont {A.}~\bibnamefont
  {Szolnoki}}\ and\ \bibinfo {author} {\bibfnamefont {M.}~\bibnamefont
  {Perc}},\ }\href {\doibase 10.1088/1367-2630/11/9/093033} {\bibfield
  {journal} {\bibinfo  {journal} {New Journal of Physics}\ }\textbf {\bibinfo
  {volume} {11}},\ \bibinfo {pages} {093033} (\bibinfo {year}
  {2009})}\BibitemShut {NoStop}%
\bibitem [{\citenamefont {Serrao}\ and\ \citenamefont
  {T{\"a}uber}(2021)}]{serrao2021stabilizing}%
  \BibitemOpen
  \bibfield  {author} {\bibinfo {author} {\bibfnamefont {S.~R.}\ \bibnamefont
  {Serrao}}\ and\ \bibinfo {author} {\bibfnamefont {U.~C.}\ \bibnamefont
  {T{\"a}uber}},\ }\href@noop {} {\bibfield  {journal} {\bibinfo  {journal}
  {The European Physical Journal B}\ }\textbf {\bibinfo {volume} {94}},\
  \bibinfo {pages} {175} (\bibinfo {year} {2021})}\BibitemShut {NoStop}%
\bibitem [{\citenamefont {Baert}\ \emph {et~al.}()\citenamefont {Baert},
  \citenamefont {De~Laender}, \citenamefont {Sabbe},\ and\ \citenamefont
  {Janssen}}]{https://doi.org/10.1002/ecy.1601}%
  \BibitemOpen
  \bibfield  {author} {\bibinfo {author} {\bibfnamefont {J.~M.}\ \bibnamefont
  {Baert}}, \bibinfo {author} {\bibfnamefont {F.}~\bibnamefont {De~Laender}},
  \bibinfo {author} {\bibfnamefont {K.}~\bibnamefont {Sabbe}}, \ and\ \bibinfo
  {author} {\bibfnamefont {C.~R.}\ \bibnamefont {Janssen}},\ }\href {\doibase
  https://doi.org/10.1002/ecy.1601} {\bibfield  {journal} {\bibinfo  {journal}
  {Ecology}\ }\textbf {\bibinfo {volume} {97}},\ \bibinfo {pages}
  {3433}}\BibitemShut {NoStop}%
\bibitem [{\citenamefont {Hillebrand}\ and\ \citenamefont
  {Kunze}(2020)}]{https://doi.org/10.1111/ele.13457}%
  \BibitemOpen
  \bibfield  {author} {\bibinfo {author} {\bibfnamefont {H.}~\bibnamefont
  {Hillebrand}}\ and\ \bibinfo {author} {\bibfnamefont {C.}~\bibnamefont
  {Kunze}},\ }\href@noop {} {\bibfield  {journal} {\bibinfo  {journal} {Ecology
  Letters}\ }\textbf {\bibinfo {volume} {23}},\ \bibinfo {pages} {575}
  (\bibinfo {year} {2020})}\BibitemShut {NoStop}%
\bibitem [{\citenamefont {Baert}\ \emph {et~al.}(2016)\citenamefont {Baert},
  \citenamefont {Janssen}, \citenamefont {Sabbe},\ and\ \citenamefont
  {De~Laender}}]{baert2016per}%
  \BibitemOpen
  \bibfield  {author} {\bibinfo {author} {\bibfnamefont {J.~M.}\ \bibnamefont
  {Baert}}, \bibinfo {author} {\bibfnamefont {C.~R.}\ \bibnamefont {Janssen}},
  \bibinfo {author} {\bibfnamefont {K.}~\bibnamefont {Sabbe}}, \ and\ \bibinfo
  {author} {\bibfnamefont {F.}~\bibnamefont {De~Laender}},\ }\href@noop {}
  {\bibfield  {journal} {\bibinfo  {journal} {Nature Communications}\ }\textbf
  {\bibinfo {volume} {7}},\ \bibinfo {pages} {12486} (\bibinfo {year}
  {2016})}\BibitemShut {NoStop}%
\bibitem [{\citenamefont {De~Laender}\ \emph {et~al.}(2016)\citenamefont
  {De~Laender}, \citenamefont {Rohr}, \citenamefont {Ashauer}, \citenamefont
  {Baird}, \citenamefont {Berger}, \citenamefont {Eisenhauer}, \citenamefont
  {Grimm}, \citenamefont {Hommen}, \citenamefont {Maltby}, \citenamefont
  {Meli{\`a}n} \emph {et~al.}}]{de2016reintroducing}%
  \BibitemOpen
  \bibfield  {author} {\bibinfo {author} {\bibfnamefont {F.}~\bibnamefont
  {De~Laender}}, \bibinfo {author} {\bibfnamefont {J.~R.}\ \bibnamefont
  {Rohr}}, \bibinfo {author} {\bibfnamefont {R.}~\bibnamefont {Ashauer}},
  \bibinfo {author} {\bibfnamefont {D.~J.}\ \bibnamefont {Baird}}, \bibinfo
  {author} {\bibfnamefont {U.}~\bibnamefont {Berger}}, \bibinfo {author}
  {\bibfnamefont {N.}~\bibnamefont {Eisenhauer}}, \bibinfo {author}
  {\bibfnamefont {V.}~\bibnamefont {Grimm}}, \bibinfo {author} {\bibfnamefont
  {U.}~\bibnamefont {Hommen}}, \bibinfo {author} {\bibfnamefont
  {L.}~\bibnamefont {Maltby}}, \bibinfo {author} {\bibfnamefont {C.~J.}\
  \bibnamefont {Meli{\`a}n}},  \emph {et~al.},\ }\href@noop {} {\bibfield
  {journal} {\bibinfo  {journal} {Trends in Ecology \& Evolution}\ }\textbf
  {\bibinfo {volume} {31}},\ \bibinfo {pages} {905} (\bibinfo {year}
  {2016})}\BibitemShut {NoStop}%
\bibitem [{\citenamefont {Chatterjee}\ \emph {et~al.}(2023)\citenamefont
  {Chatterjee}, \citenamefont {De}, \citenamefont {Hens}, \citenamefont {Dana},
  \citenamefont {Kapitaniak},\ and\ \citenamefont
  {Bhattacharyya}}]{chatterjee2023response}%
  \BibitemOpen
  \bibfield  {author} {\bibinfo {author} {\bibfnamefont {S.}~\bibnamefont
  {Chatterjee}}, \bibinfo {author} {\bibfnamefont {R.}~\bibnamefont {De}},
  \bibinfo {author} {\bibfnamefont {C.}~\bibnamefont {Hens}}, \bibinfo {author}
  {\bibfnamefont {S.~K.}\ \bibnamefont {Dana}}, \bibinfo {author}
  {\bibfnamefont {T.}~\bibnamefont {Kapitaniak}}, \ and\ \bibinfo {author}
  {\bibfnamefont {S.}~\bibnamefont {Bhattacharyya}},\ }\href@noop {} {\bibfield
   {journal} {\bibinfo  {journal} {Scientific Reports}\ }\textbf {\bibinfo
  {volume} {13}},\ \bibinfo {pages} {20740} (\bibinfo {year}
  {2023})}\BibitemShut {NoStop}%
\bibitem [{\citenamefont {Chen}\ \emph {et~al.}(2024)\citenamefont {Chen},
  \citenamefont {Wang},\ and\ \citenamefont {Liu}}]{chen2024stability}%
  \BibitemOpen
  \bibfield  {author} {\bibinfo {author} {\bibfnamefont {C.}~\bibnamefont
  {Chen}}, \bibinfo {author} {\bibfnamefont {X.-W.}\ \bibnamefont {Wang}}, \
  and\ \bibinfo {author} {\bibfnamefont {Y.-Y.}\ \bibnamefont {Liu}},\
  }\href@noop {} {\bibfield  {journal} {\bibinfo  {journal} {Physics reports}\
  }\textbf {\bibinfo {volume} {1088}},\ \bibinfo {pages} {1} (\bibinfo {year}
  {2024})}\BibitemShut {NoStop}%
\bibitem [{\citenamefont {Dai}\ \emph {et~al.}(2025)\citenamefont {Dai},
  \citenamefont {Wang}, \citenamefont {Dai},\ and\ \citenamefont
  {Shi}}]{dai2025chaos}%
  \BibitemOpen
  \bibfield  {author} {\bibinfo {author} {\bibfnamefont {H.}~\bibnamefont
  {Dai}}, \bibinfo {author} {\bibfnamefont {X.}~\bibnamefont {Wang}}, \bibinfo
  {author} {\bibfnamefont {X.}~\bibnamefont {Dai}}, \ and\ \bibinfo {author}
  {\bibfnamefont {L.}~\bibnamefont {Shi}},\ }\href@noop {} {\bibfield
  {journal} {\bibinfo  {journal} {Chaos: An Interdisciplinary Journal of
  Nonlinear Science}\ }\textbf {\bibinfo {volume} {35}},\ \bibinfo {pages}
  {093147} (\bibinfo {year} {2025})}\BibitemShut {NoStop}%
\bibitem [{\citenamefont {Reuter}\ \emph {et~al.}(2010)\citenamefont {Reuter},
  \citenamefont {Jopp}, \citenamefont {Blanco-Moreno}, \citenamefont
  {Damgaard}, \citenamefont {Matsinos},\ and\ \citenamefont
  {DeAngelis}}]{reuter2010ecological}%
  \BibitemOpen
  \bibfield  {author} {\bibinfo {author} {\bibfnamefont {H.}~\bibnamefont
  {Reuter}}, \bibinfo {author} {\bibfnamefont {F.}~\bibnamefont {Jopp}},
  \bibinfo {author} {\bibfnamefont {J.~M.}\ \bibnamefont {Blanco-Moreno}},
  \bibinfo {author} {\bibfnamefont {C.}~\bibnamefont {Damgaard}}, \bibinfo
  {author} {\bibfnamefont {Y.}~\bibnamefont {Matsinos}}, \ and\ \bibinfo
  {author} {\bibfnamefont {D.~L.}\ \bibnamefont {DeAngelis}},\ }\href@noop {}
  {\bibfield  {journal} {\bibinfo  {journal} {Basic and Applied Ecology}\
  }\textbf {\bibinfo {volume} {11}},\ \bibinfo {pages} {572} (\bibinfo {year}
  {2010})}\BibitemShut {NoStop}%
\bibitem [{\citenamefont {Miller~III}(2008)}]{miller2008hierarchical}%
  \BibitemOpen
  \bibfield  {author} {\bibinfo {author} {\bibfnamefont {W.}~\bibnamefont
  {Miller~III}},\ }\href@noop {} {\bibfield  {journal} {\bibinfo  {journal}
  {Evolution: Education and Outreach}\ }\textbf {\bibinfo {volume} {1}},\
  \bibinfo {pages} {16} (\bibinfo {year} {2008})}\BibitemShut {NoStop}%
\bibitem [{\citenamefont {Maia}\ and\ \citenamefont
  {Guimaraes~Jr}(2024)}]{maia2024hierarchical}%
  \BibitemOpen
  \bibfield  {author} {\bibinfo {author} {\bibfnamefont {K.~P.}\ \bibnamefont
  {Maia}}\ and\ \bibinfo {author} {\bibfnamefont {P.~R.}\ \bibnamefont
  {Guimaraes~Jr}},\ }\href@noop {} {\bibfield  {journal} {\bibinfo  {journal}
  {Ecology Letters}\ }\textbf {\bibinfo {volume} {27}},\ \bibinfo {pages}
  {e14501} (\bibinfo {year} {2024})}\BibitemShut {NoStop}%
\bibitem [{\citenamefont {Yang}\ \emph {et~al.}(2025)\citenamefont {Yang},
  \citenamefont {Hong}, \citenamefont {Kim},\ and\ \citenamefont
  {Park}}]{yang2025understanding}%
  \BibitemOpen
  \bibfield  {author} {\bibinfo {author} {\bibfnamefont {R.~K.}\ \bibnamefont
  {Yang}}, \bibinfo {author} {\bibfnamefont {S.}~\bibnamefont {Hong}}, \bibinfo
  {author} {\bibfnamefont {S.}~\bibnamefont {Kim}}, \ and\ \bibinfo {author}
  {\bibfnamefont {J.}~\bibnamefont {Park}},\ }\href@noop {} {\bibfield
  {journal} {\bibinfo  {journal} {Chaos: An Interdisciplinary Journal of
  Nonlinear Science}\ }\textbf {\bibinfo {volume} {35}} (\bibinfo {year}
  {2025})}\BibitemShut {NoStop}%
\bibitem [{\citenamefont {Szolnoki}\ \emph {et~al.}(2014)\citenamefont
  {Szolnoki}, \citenamefont {Mobilia}, \citenamefont {Jiang}, \citenamefont
  {Szczesny}, \citenamefont {Rucklidge},\ and\ \citenamefont
  {Perc}}]{szolnoki2014cyclic}%
  \BibitemOpen
  \bibfield  {author} {\bibinfo {author} {\bibfnamefont {A.}~\bibnamefont
  {Szolnoki}}, \bibinfo {author} {\bibfnamefont {M.}~\bibnamefont {Mobilia}},
  \bibinfo {author} {\bibfnamefont {L.-L.}\ \bibnamefont {Jiang}}, \bibinfo
  {author} {\bibfnamefont {B.}~\bibnamefont {Szczesny}}, \bibinfo {author}
  {\bibfnamefont {A.~M.}\ \bibnamefont {Rucklidge}}, \ and\ \bibinfo {author}
  {\bibfnamefont {M.}~\bibnamefont {Perc}},\ }\href@noop {} {\bibfield
  {journal} {\bibinfo  {journal} {Journal of the Royal Society Interface}\
  }\textbf {\bibinfo {volume} {11}} (\bibinfo {year} {2014})}\BibitemShut
  {NoStop}%
\bibitem [{\citenamefont {Nag~Chowdhury}\ \emph
  {et~al.}(2021{\natexlab{a}})\citenamefont {Nag~Chowdhury}, \citenamefont
  {Kundu}, \citenamefont {Perc},\ and\ \citenamefont
  {Ghosh}}]{chowdhury2021complex}%
  \BibitemOpen
  \bibfield  {author} {\bibinfo {author} {\bibfnamefont {S.}~\bibnamefont
  {Nag~Chowdhury}}, \bibinfo {author} {\bibfnamefont {S.}~\bibnamefont
  {Kundu}}, \bibinfo {author} {\bibfnamefont {M.}~\bibnamefont {Perc}}, \ and\
  \bibinfo {author} {\bibfnamefont {D.}~\bibnamefont {Ghosh}},\ }\href@noop {}
  {\bibfield  {journal} {\bibinfo  {journal} {Proceedings of the Royal Society
  A: Mathematical, Physical and Engineering Sciences}\ }\textbf {\bibinfo
  {volume} {477}} (\bibinfo {year} {2021}{\natexlab{a}})}\BibitemShut {NoStop}%
\bibitem [{\citenamefont {Nag~Chowdhury}\ \emph {et~al.}(2023)\citenamefont
  {Nag~Chowdhury}, \citenamefont {Banerjee}, \citenamefont {Perc},\ and\
  \citenamefont {Ghosh}}]{chowdhury2023eco}%
  \BibitemOpen
  \bibfield  {author} {\bibinfo {author} {\bibfnamefont {S.}~\bibnamefont
  {Nag~Chowdhury}}, \bibinfo {author} {\bibfnamefont {J.}~\bibnamefont
  {Banerjee}}, \bibinfo {author} {\bibfnamefont {M.}~\bibnamefont {Perc}}, \
  and\ \bibinfo {author} {\bibfnamefont {D.}~\bibnamefont {Ghosh}},\
  }\href@noop {} {\bibfield  {journal} {\bibinfo  {journal} {Journal of
  Theoretical Biology}\ }\textbf {\bibinfo {volume} {564}},\ \bibinfo {pages}
  {111446} (\bibinfo {year} {2023})}\BibitemShut {NoStop}%
\bibitem [{\citenamefont {Nag~Chowdhury}\ \emph
  {et~al.}(2021{\natexlab{b}})\citenamefont {Nag~Chowdhury}, \citenamefont
  {Kundu}, \citenamefont {Banerjee}, \citenamefont {Perc},\ and\ \citenamefont
  {Ghosh}}]{chowdhury2021eco}%
  \BibitemOpen
  \bibfield  {author} {\bibinfo {author} {\bibfnamefont {S.}~\bibnamefont
  {Nag~Chowdhury}}, \bibinfo {author} {\bibfnamefont {S.}~\bibnamefont
  {Kundu}}, \bibinfo {author} {\bibfnamefont {J.}~\bibnamefont {Banerjee}},
  \bibinfo {author} {\bibfnamefont {M.}~\bibnamefont {Perc}}, \ and\ \bibinfo
  {author} {\bibfnamefont {D.}~\bibnamefont {Ghosh}},\ }\href@noop {}
  {\bibfield  {journal} {\bibinfo  {journal} {Journal of Theoretical Biology}\
  }\textbf {\bibinfo {volume} {518}},\ \bibinfo {pages} {110606} (\bibinfo
  {year} {2021}{\natexlab{b}})}\BibitemShut {NoStop}%
\bibitem [{\citenamefont {Roy}\ \emph {et~al.}(2023)\citenamefont {Roy},
  \citenamefont {Nag~Chowdhury}, \citenamefont {Kundu}, \citenamefont {Sar},
  \citenamefont {Banerjee}, \citenamefont {Rakshit}, \citenamefont {Mali},
  \citenamefont {Perc},\ and\ \citenamefont {Ghosh}}]{roy2023time}%
  \BibitemOpen
  \bibfield  {author} {\bibinfo {author} {\bibfnamefont {S.}~\bibnamefont
  {Roy}}, \bibinfo {author} {\bibfnamefont {S.}~\bibnamefont {Nag~Chowdhury}},
  \bibinfo {author} {\bibfnamefont {S.}~\bibnamefont {Kundu}}, \bibinfo
  {author} {\bibfnamefont {G.~K.}\ \bibnamefont {Sar}}, \bibinfo {author}
  {\bibfnamefont {J.}~\bibnamefont {Banerjee}}, \bibinfo {author}
  {\bibfnamefont {B.}~\bibnamefont {Rakshit}}, \bibinfo {author} {\bibfnamefont
  {P.~C.}\ \bibnamefont {Mali}}, \bibinfo {author} {\bibfnamefont
  {M.}~\bibnamefont {Perc}}, \ and\ \bibinfo {author} {\bibfnamefont
  {D.}~\bibnamefont {Ghosh}},\ }\href@noop {} {\bibfield  {journal} {\bibinfo
  {journal} {Scientific Reports}\ }\textbf {\bibinfo {volume} {13}},\ \bibinfo
  {pages} {14331} (\bibinfo {year} {2023})}\BibitemShut {NoStop}%
\bibitem [{\citenamefont {Park}\ \emph {et~al.}(2013)\citenamefont {Park},
  \citenamefont {Do}, \citenamefont {Huang},\ and\ \citenamefont
  {Lai}}]{park2013persistent}%
  \BibitemOpen
  \bibfield  {author} {\bibinfo {author} {\bibfnamefont {J.}~\bibnamefont
  {Park}}, \bibinfo {author} {\bibfnamefont {Y.}~\bibnamefont {Do}}, \bibinfo
  {author} {\bibfnamefont {Z.-G.}\ \bibnamefont {Huang}}, \ and\ \bibinfo
  {author} {\bibfnamefont {Y.-C.}\ \bibnamefont {Lai}},\ }\href@noop {}
  {\bibfield  {journal} {\bibinfo  {journal} {Chaos: An Interdisciplinary
  Journal of Nonlinear Science}\ }\textbf {\bibinfo {volume} {23}} (\bibinfo
  {year} {2013})}\BibitemShut {NoStop}%
\bibitem [{\citenamefont {Park}\ \emph {et~al.}(2018)\citenamefont {Park},
  \citenamefont {Do},\ and\ \citenamefont {Jang}}]{park2018multistability}%
  \BibitemOpen
  \bibfield  {author} {\bibinfo {author} {\bibfnamefont {J.}~\bibnamefont
  {Park}}, \bibinfo {author} {\bibfnamefont {Y.}~\bibnamefont {Do}}, \ and\
  \bibinfo {author} {\bibfnamefont {B.}~\bibnamefont {Jang}},\ }\href@noop {}
  {\bibfield  {journal} {\bibinfo  {journal} {Chaos: An Interdisciplinary
  Journal of Nonlinear Science}\ }\textbf {\bibinfo {volume} {28}} (\bibinfo
  {year} {2018})}\BibitemShut {NoStop}%
\bibitem [{\citenamefont {Miller}\ and\ \citenamefont
  {Max}(2025)}]{miller2025multispecies}%
  \BibitemOpen
  \bibfield  {author} {\bibinfo {author} {\bibfnamefont {Z.~R.}\ \bibnamefont
  {Miller}}\ and\ \bibinfo {author} {\bibfnamefont {D.}~\bibnamefont {Max}},\
  }\href@noop {} {\bibfield  {journal} {\bibinfo  {journal} {Ecology Letters}\
  }\textbf {\bibinfo {volume} {28}},\ \bibinfo {pages} {e70206} (\bibinfo
  {year} {2025})}\BibitemShut {NoStop}%
\bibitem [{\citenamefont {Hastings}(1980)}]{hastings1980disturbance}%
  \BibitemOpen
  \bibfield  {author} {\bibinfo {author} {\bibfnamefont {A.}~\bibnamefont
  {Hastings}},\ }\href@noop {} {\bibfield  {journal} {\bibinfo  {journal}
  {Theoretical Population Biology}\ }\textbf {\bibinfo {volume} {18}},\
  \bibinfo {pages} {363} (\bibinfo {year} {1980})}\BibitemShut {NoStop}%
\end{thebibliography}%
\bibliographystyle{apsrev4-1}

\end{document}